\documentclass[11pt,authoryear]{elsarticle}
\usepackage{mydef}
\usepackage{tabularx}

\newcommand{\p}{\partial}
\def\hatr{{\hat r}}
\def\hatR{{\hat R}}
\def\hatB{{\hat B}}
\def\m1{\mathbbm{1}}
\def\calR{{\mathcal R}}
\def\calF{{\mathcal F}}
\def\QM{{\mathbb Q}}
\def\Xs{{\mathbf{X}_s}}
\def\Xsm{{\mathbf{X}_{s^-}}}

\graphicspath{{./figs/}}

\begin{document}

\begin{frontmatter}

\title{Influence of jump-at-default in IR and FX on Quanto CDS prices}

\author[nyu]{A.~Itkin\corref{cor1}}
\ead{aitkin@nyu.edu}

\author[uppsala]{V. Shcherbakov}
\ead{victor.shcherbakov@it.uu.se}

\author[hsbs]{A. Veygman}
\ead{veygmana@gmail.com}

\address[nyu]{Tandon School of Engineering, New York University, \\
12 Metro Tech Center, RH 517E, Brooklyn NY 11201, USA}

\address[uppsala]{Department of Information Technology, Division of Scientific Computing, \\
Box 337, 751 05 Uppsala, Sweden}

\address[hsbs]{HSBS, New York, USA}

\begin{abstract}
We propose a new model for pricing Quanto CDS and risky bonds. The model operates with four stochastic factors, namely: hazard rate, foreign exchange rate, domestic interest rate, and foreign interest rate, and also allows for jumps-at-default in the FX and foreign interest rates. Corresponding systems of PDEs are derived similar to how this is done in \cite{BieleckiPDE2005}. A localized version of the RBF partition of unity method is used to solve these 4D PDEs. The results of our numerical experiments presented in the paper qualitatively explain the discrepancies observed in the marked values of CDS spreads traded in domestic and foreign economies.
\end{abstract}


\end{frontmatter}

%%%%%%%%%%%%%%%%%%%%%%%%%%%%%%%%%%%%%%%%%%%%%%%%%
\section{Introduction}
%%%%%%%%%%%%%%%%%%%%%%%%%%%%%%%%%%%%%%%%%%%%%%%%%
Quanto CDS is a credit default swap (CDS) with a special feature that the swap premium payments, and/or the cashflows in the case of default, are done in a different currency to that of the reference asset. A typical example would be a CDS that has its reference as a dollar-denominated bond for which the premium of the swap is payable in euros. And  in case of default the payment equals the recovery rate on the dollar bond payable in euros. In other words, this CDS is written on a dollar bond, while its premium is payable in euros. These contracts are widely used to hedge holdings in bonds or bank loans that are denominated in a foreign currency (other than the investor’s home currency).

As mentioned in \cite{Citi2010},  this product enables investors to take views on joint spread and FX moves with value a function of spread, the FX rate and FX volatility. Given the increased correlation between FX moves and credit spreads, interest in this product has increased recently, although like recovery swaps, it is still rather a niche market.

A Quanto CDS is quoted as a spread between the standard CDS and that of a different currency, and they are available for different maturities. For instance, one can observe that CDS on European sovereigns are usually traded in US dollars. That is because in case of default a euro-denominated credit protection would significantly drop down reflecting the default of the corresponding economy. So, the term structure of Quanto CDS  tells us how financial markets view the likelihood of a foreign default and associated currency devaluations at different horizons, see e.g., discussion in \cite{ACS2017} and references therein.

In Fig.~\ref{histSpread} historical time series of some European 5Y sovereign CDS traded in USD are presented for the period from 2006 to 2015. It can be seen that these spreads reach their maximum around 2011, and then drop down by factors 2-5 to their current level. However, since high levels of the spreads have been recorded, later in this paper when choosing test parameters of our numerical experiments we will look at the cases corresponding just to the period of raised spreads around 2011.

\begin{figure}[!t]
\centering
\includegraphics[width=\textwidth]{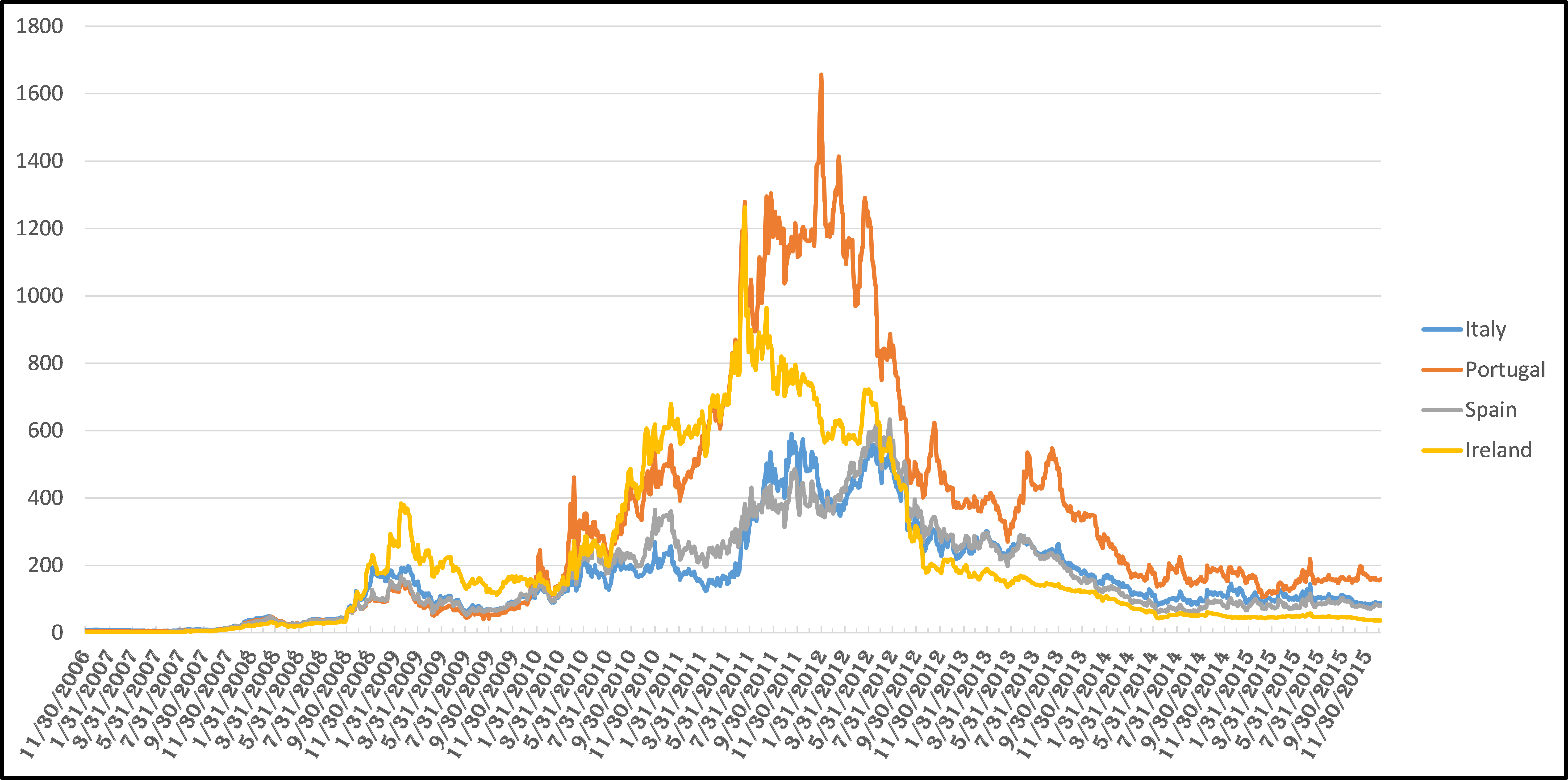}
\caption{Historical time-series of some European 5Y sovereign CDS traded in USD (Markit).}
\label{histSpread}
\end{figure}

As far as the value of the Quanto CDS spread is concerned, there are various data in the literature. For instance,  in \cite{ACS2017} the term structure of spreads, defined as the difference between the USD and EUR denominated CDS spreads, is presented for six Eurozone countries: Germany , Belgium, France, Ireland, Italy, Portugal and for maturities 3, 5, 7, 10, and 15 years relative to the 1 year Quanto spread. This difference could reach 30 bps at the time horizon 15 years (France, Ireland). In \cite{Simon2015} the 5 years Quanto CDS spreads are presented for Germany, Italy and France over the period from 2004 to 2013, which, e.g., for Italy could reach 500 bps in 2012. The results presented in \cite{Brigo} indicate a significant basis across domestic and foreign CDS quotes. For instance, for Italy, a USD CDS spread quote of 440 bps can translate into a EUR quote of 350 bps in the middle of the Euro-debt crisis in the first week of May 2012. More recently, from June 2013, the basis spreads between the EUR quotes and the USD quotes are in the range of around 40 bps.

Quanto effects drew a lot of attention on a modelling side. Various aspects of the problem were under investigation including the relationship between sovereign credit and currency risks, the pricing of sovereign CDS, the impact of contagion on credit risk, see survey in \cite{ACS2017} and references therein. But in this paper our particular attention will be directed to pricing Quanto CDS, or, more rigorously, to determining and testing an appropriate framework that provides a reasonable explanation of these effects from a mathematical finance point of view. Our approach is close to that in \cite{Brigo} where a model of Quanto CDS is built based on the reduced form model for credit risk. Within this setting the default time is modeled as a Cox process with explicit diffusion dynamics for default intensity/hazard rate and exponential jump to default, similar to the approach of \cite{ES2006, Mohammadi2006}. But what is more important, \cite{Brigo} introduce an explicit jump-at-default in the FX dynamics. Then they show that this provides a more effective way to model credit/FX dependency as the results of simulation are able to explain the observed basis spreads during the Euro-debt crisis. In contrast, taking into account the instantaneous correlation between the driving Brownian motions of the default intensity and the FX rate alone is not sufficient for doing so.

However, in \cite{Brigo} only deterministic domestic and foreign interest rates (IR) were considered. While it could be important to extend this approach by relaxing this assumption and letting the rates be stochastic. It would be more important to account not just for the jump-at-default in the FX rate, but also for a simultaneous jump-at-default in the interest rate of the defaulted country. Relevant data on the subject could be found, e.g., in \cite{Catao2015}. This investigation shows that the interest rate premium on past default has been underestimated. This is partly due to narrower credit history indicators and, crucially, to the narrower data coverage of previous studies. Once this correction is made for these problems, a sizeable and persistent default premium emerges, and one which rises on the duration of the default. This means that the longer a country stays in default the higher premium it will pay once it resumes borrowing from private capital markets.

Another example is given in \cite{Katselas2010}. He provides a plot of the overnight interbank cash rate as quoted by the Reserve Bank of Australia for the period starting on 4 January 2000 and finishing on 31 December 2009. This rate serves as an approximation to the risk-free short rate applicable to borrowing/lending in Australia, and the plot indicates that not only are jumps evident in the short rate, but that a pure jump process may act as a suitable model for short rates. This observation prompted, e.g.,
\cite{Borovkov2003}, to consider using a marked Poisson point process to model the short rate as a pure jump process.

Therefore, in this paper we extend the framework of \cite{Brigo} by introducing stochastic interest rates and account for jump-at-default in both FX and foreign (defaulted) interest rates. Our goal is to compare contribution of both jumps into the value of Quanto CDS spread. As this problem has four stochastic drivers, plus time, we show that the corresponding CDS price solves a four-dimensional partial differential equation. It is well-known that this dimensionality is such that finite-difference method already immensely suffer from the curse of dimensionality, while using Monte Carlo methods is too computationally expensive. Therefore, here we used another method, namely, a radial basis function (RBF) method, which has already demonstrated its efficiency when solving various problems of intermediate ($10 > d > 3$) dimensionality including those in mathematical finance, see, e.g.,  \cite{YCHon3,Fasshauer2,Pettersson}, thanks to its high order convergence. The latter allows for obtaining a high resolution scheme using just a few discretization nodes. In particular, in this paper a localized version of the RBF method is used. It is based on the partition of unity method (or RBF-PUM). The partition of unity was originally introduced by \cite{Melenk} for finite element methods, and later adapted for the RBF methods by several authors, \cite{Safdari,Shcherbakov}. This approach enables a significant reduction in the number of non-zero elements that remain in the coefficient matrix, hence, lowering the computational intensity required for solving the system.

The rest of the paper is organized as follows. In Section~\ref{model} we describe our model, and derive the main partial differential equation (PDE) for the risky bond price under this model. In Section~\ref{modelJumps} we extend this framework by adding jumps-at-default into the dynamics of the FX and foreign  (defaulted) interest rates. Again, the main PDE is derived for the risky bond (the detailed derivation is given in Appendix). The connection of this price with the prices of the Quanto CDS is established in Section~\ref{bond2cds}. In Section~\ref{numMethod} the RBF-PUM method is described in detail. In Section~\ref{experiments} we present numerical results of our experiments with this model and discussion of the observed effects. Finally, Section~\ref{sec:Conclusion} concludes the paper.

%%%%%%%%%%%%%%%%%%%%%%%%%%%%%%%%%%%%%%%%%%%%%%%%%
\section{Model} \label{model}
%%%%%%%%%%%%%%%%%%%%%%%%%%%%%%%%%%%%%%%%%%%%%%%%%

We begin describing our model by giving some useful definitions which are heavily utilized throughout the rest of the paper.

By the {\it domestic currency} or the {\it liquid currency} we denote the most liquidly traded currency among  all contractual currencies. In what follows this is the US dollar (USD).

The other contractual currency we denote as {\it contractual} or {\it foreign currency}. In this paper it can be both USD and EUR. The premium and protection leg payments are settled in this currency.

Since in this paper we focus on pricing credit default swap (CDS) contracts, it is assumed that their market quotes are available in both domestic and foreign currencies. Let us denote these prices as $\mathrm{CDS}_d$ and $\mathrm{CDS}_f$ respectively. If so, every price $\mathrm{CDS}_f$ expressed in the foreign currency can be translated into the corresponding price in the domestic currency if the exchange rate $Z_t$ for two currencies is provided by the market. In other words, the theoretical price of the CDS contract in the foreign currency would be $Z_t \mathrm{CDS}_d$. However, it is known that the market demonstrates a spread $\mathrm{CDS}_f - Z_t \mathrm{CDS}_d$ which could reach hundreds of bps, \cite{Brigo}. Hence, the availability of the market quotes on CDS contracts in both currencies together with the corresponding exchange rates allows one to capture these spreads.

We continue our description by considering a framework where all underlying stochastic processes do not experience a jump-at-default except the default process itself. So, this is similar to what is presented in \cite{Brigo} with an exception that the interest rates in our model are stochastic. This will then be generalized with the allowance for jumps-at-default in other processes in Section~\ref{modelJumps}.

\subsection{Simple jump-at-default framework}

Below we chose the risk neutral probability measure $\mathbb{Q}$ corresponding to the domestic (liquid) currency money market. Also, by $\mathbb{E}_t[\,\cdot\,]$ we denote the expectation conditioned on the information received by time $t$, i.e. $\mathbb{E}[\,\cdot\, | \mathcal{F}_t]$.

Consider two money markets: $B_t$ associated with the domestic currency (USD), and $\hatB_t$ associated with the foreign currency (EUR), where $t\geq 0$ is the calendar time. We assume that the dynamics of the two money market accounts are given by
\begin{align} \label{mmDyn}
dB_t & = R_t B_t dt, \quad B_0=1,\\
d\hatB_t & = \hatR_t \hatB_t dt, \quad \hatB_0=1, \nonumber
\end{align}
\noindent where the stochastic interest rates $R_t, \hatR_t$ follow the Cox-Ingersoll-Ross (CIR) process, \cite{cir:85}
\begin{align} \label{dynR}
  dR_t &= a(b-R_t)dt + \sigma_r \sqrt{R_t} dW_t^{(1)},  \quad R_0=r,\\
  d\hatR_t &= \hat a(\hat b - \hatR_t) dt + \sigma_{\hat{r}} \sqrt{\hatR_t}dW_t^{(2)}, \quad \hatR_0=\hatr. \nonumber
\end{align}
Here $a, \hat a$ are the mean-reversion rates, $b, \hat b$ are the mean-reversion levels, $\sigma_r, \sigma_{\hat{r}}$ are the volatilities, and $W_t^{(1)}, W_t^{(2)}$ are the Brownian motions. Without loss of generality, further we assume $a, \hat a, b, \hat b, \sigma_r, \sigma_{\hatr}$ to be constant. This assumption can be easily relaxed.

We assume that the exchange rate $Z_t$ of the two currencies is also stochastic, and its  dynamics is driven by the following stochastic differential equation (SDE)
\begin{equation} \label{Z}
dZ_t = \mu_z Z_t dt + \sigma_z Z_t dW_t^{(3)}, \quad Z_0=z,
\end{equation}
\noindent where $\mu_z, \sigma_z$ are the corresponding drift and volatility, and $W_t^{(3)}$ is another Brownian motion. From the financial point of view $Z_t$ denotes the amount of domestic currency one has to pay to buy one unit of foreign currency. Loosely speaking, this means that 1 euro could be exchanged for $Z_t$ US dollars.

As the underlying security of a CDS contract is a risky bond, we need a model of a credit risk  implied by the bond. For modeling the credit risk we use a reduced form model approach, see e.g., \cite{jarrow/turnbull:95, DuffieSingleton99, Bielecki2004, jarrow2003robust} and references therein. We define the hazard rate $\lambda_t$ to be a stochastic process given by
\begin{equation} \label{lambda}
\lambda_t = e^{Y_t}, \quad t \ge 0,
\end{equation}
\noindent where $Y_t$ follows the Ornstein-Uhlenbeck process defined by the SDE
\begin{equation} \label{Y}
dY_t = \kappa(\theta-Y_t)dt + \sigma_y dW_t^{(4)}, \quad Y_0=y, \\
\end{equation}
\noindent with $\kappa, \theta, \sigma_y$ to be the corresponding mean-reversion rate, mean-reversion level and volatility, and $W_t^{(4)}$ to be another Brownian motion. Both $Z_t$ and $\lambda_t$  are defined and calibrated in the domestic measure.

We assume all Brownian motions $W_t^{(i)}, \ i \in [1,4]$ to be dependent, and this dependence can be specified through the instantaneous correlation $\rho$ between each pair of the Brownian motions, i.e., $<d W_t^{(i)}, d W_t^{(j)}> = \rho_{ij} dt$. Hence, the whole correlation matrix in our model is
\begin{equation}
  \cal P =
 \begin{bmatrix}
    1 & \rho_{r\hatr} & \rho_{rz} & \rho_{ry} \\
    \rho_{\hatr r} & 1 & \rho_{\hatr z} & \rho_{\hatr y} \\
    \rho_{zr} & \rho_{z \hatr} & 1 & \rho_{z y} \\
    \rho_{yr} & \rho_{y \hatr} & \rho_{yz} & 1 \\
  \end{bmatrix},
\end{equation}
\noindent where all correlations $|\rho_{ij}| \le 1, \ i,j \in [r,\hatr, z, y]$ are assumed to be constant.

Finally, we define the default process $(D_t, \ t \ge 0)$ as
\begin{equation} \label{defProc}
D_t = {\bf 1}_{\tau \le t},
\end{equation}
\noindent where $\tau$ is the default time of the reference entity. In order to exclude trivial cases, we assume that $\QM(\tau > 0) = 1$, and  $\QM(\tau \le T) > 0$.

\subsection{Jumps-at-default in FX and foreign IR} \label{modelJumps}

In this section we extend the above described framework by assuming the value of the foreign currency as well as the foreign interest rate to experience a jump at the default time.

As shown in \cite{Brigo} and mentioned in the introduction, including jump-at-default into the FX rate provides a more effective way of modeling the credit/FX dependency than the instantaneous correlations imposed among the driving Brownian motions of default
intensity and FX rates. Moreover, the authors claim that it is not possible to explain the observed basis spreads during the Euro-debt crisis by using the latter mechanism alone.

However, looking at historical time-series, an existence of jump-at-default in the foreign interest rate could also be justified, especially in case when sovereign obligations are in question. For example, after the default of Russia in 1998, the Russian ruble lost about $75$\% of its value within $1.5$ months, which in turn resulted in a jump of the corresponding FX rates. On the other hand, the jump in the interest rate can be even more pronounced since the default also lowers the creditability and dramatically increases the cost of borrowing. For the above mentioned example of the Russian crisis of 1998, the short interest rate grew from $20$\% in April 1998 to $120$\% in August 1998.

Therefore, it would be interesting to see a relative contribution of each jump into the value of the Quanto CDS spread.

To add jumps to the dynamics of the FX rate in \eqref{Z}, we follow \cite{Brigo, BieleckiPDE2005} who assume that at the time of default the FX rate experiences a single jump which is proportional to the current rate level, i.e.
\begin{equation} \label{jumpZ}
d Z_t = \gamma_z Z_{t^-} d M_t,
\end{equation}
\noindent where $\gamma_z \in [-1,\infty)$ \footnote{This is to prevent $Z_t$ to be negative, \cite{BieleckiPDE2005}.} is a devaluation/revaluation parameter.

The hazard process $\Gamma_t$ of a random time $\tau$ with respect to a reference filtration is defined through the equality $e^{-\Gamma_t} = 1 - \QM\{\tau \le t|\calF_t\}$. It is well known that if the hazard process $\Gamma_t$ of $\tau$ is absolutely continuous, so
\begin{equation} \label{hazard}
\Gamma_t = \int_0^t (1-D_s) \lambda_s ds,
\end{equation}
\noindent and increasing, then the process $M_t = D_t - \Gamma_t$ is a martingale
(which is called as the compensated martingale of the default process $D_t$) under the full filtration $\calF_t \vee {\mathcal H}_t$ with ${\mathcal H}_t$ being the filtration generated by the default process. So, $M_t$ is a martingale under $\QM$, \cite{BieleckiPDE2005}.

It can be shown that under the risk-neutral measure associated with the domestic currency, the drift $\mu_z$ is, (\cite{Brigo})
\begin{equation} \label{na-drift}
\mu_z = R_t-\hatR_t.
\end{equation}

Therefore, with the allowance for \eqref{Z}, \eqref{jumpZ} we obtain
\begin{equation} \label{dzJump}
dZ_t =  (R_t - \hatR_t) Z_t dt + \sigma_z Z_t dW_t^{(3)} + \gamma_z Z_t d M_t.
\end{equation}
Thus, $Z_t$ is a martingale under the $\mathbb{Q}$-measure with respect to $\calF_t \vee {\mathcal H}_t$ as it should be, since it is a tradable asset.

Certainly, we are more interested in the negative values of $\gamma_z$ because a default of the reference entity has to negatively impact the value of its local currency. For instance, we expect the value of EUR expressed in USD to fall if some European country defaults.

Similarly, we add jump-at-default to the stochastic process for the foreign interest rate $\hatR_t$ as
\[ d \hatR_t = \gamma_\hatr \hatR_{t^-} d D_t, \]
\noindent so \eqref{dynR} transforms to
\begin{equation} \label{rJump}
d\hatR_t = \hat a(\hat b-\hatR_t )dt + \sigma_{\hatr} \sqrt{\hatR_t}dW_t^{(2)} + \gamma_{\hatr} R_t d D_t.
\end{equation}
Here $\gamma_{\hatr} \in [-1,\infty)$ is the parameter that determines the post-default cost of borrowing. We are interested in positive values of $\gamma_{\hatr}$ as the interest rate most likely will grow after a default has occurred.
Note that $\hatR_t$ is not tradable, and so is not a martingale under the $\mathbb{Q}$-measure.

\section{Pricing zero-coupon bonds} \label{zcbPrice}

To price contingent claims where the contractual currency differs from the pricing currency, e.g., Quanto CDS, we first need to determine the price of the underlying defaultable zero-coupon bond settled in foreign currency. The bond price under the foreign money market martingale measure $\hat \QM$ reads
\begin{equation}
 \hat U_t(T) = \hat{\mathbb{E}}_t\left[ \frac{\hatB_t}{\hatB_T} \hat \Phi(T) \right],
\end{equation}
\noindent where $\hatB_t/\hatB_T = \hat B(t,T)$ is the stochastic discount factor from time $T$ to time $t$ in the foreign economy, and $\Phi(T)$ is the payoff function.
However, we are going to find this price under the domestic money market measure $\mathbb{Q}$. Hence, converting the payoff to the domestic currency and discounting by the domestic money market account yields
\begin{equation}
  U_t(T) = \mathbb{E}_t\left[ B(t,T) Z_t \hat \Phi(T) \right],
\end{equation}
\noindent where without loss of generality it is assumed that the notional amount of the contract is equal to one unit of the foreign currency. This implies the payoff function to be
\begin{equation}
  \hat \Phi(T) = \m1_{\tau>T}.
\end{equation}
Further, we assume that if this bond defaults, the recovery rate $\calR$ is paid at the time of default. Therefore, the price of a defaultable zero-coupon bond, which pays out one unit of the foreign currency in the domestic economy reads
\begin{align} \label{payoff}
U_t(T) &= \mathbb{E}_t\left[ B(t,T) Z_T \m1_{\tau>T}
+ \calR B(t,\tau) Z_\tau \m1_{\tau \le T} \right] \\
&= \mathbb{E}_t\left[ B(t,T) Z_T \m1_{\tau>T} \right]
+ \calR \int_t^T \mathbb{E}_t \left[ B(t,\nu) Z_\nu \m1_{\tau \in (\nu-d\nu,nu]} \right]  = w_t(T) + \calR \int_t^T g_t(\nu)d \nu, \nonumber \\
w_t(T) &:= \mathbb{E}_t \left[ Z_{T} B(t,T) \m1_{\tau > T} \right], \qquad
g_t(\nu) := \mathbb{E}_t \left[ B(t,\nu) Z_\nu \frac{\m1_{\tau \in (\nu-d\nu,nu]}}{d\nu} \right]. \nonumber
\end{align}

As the whole dynamics of our underlying processes is Markovian, \cite{BieleckiPDE2005}, to find the price of such an instrument we use a PDE approach, so that the defaultable bond price just solves it. This is more efficient from the computationally point of view as compared, e.g., with the Monte Carlo method, despite the resulting PDE becomes four-dimensional. We discuss various approaches to its numerical solution in Section~\ref{numMethod}.

Further, conditioning on $R_t = r, \hatR_t = \hatr, Z_t = z, Y_t = y, D_t = d$, and using the approach of \cite{BieleckiPDE2005} (see Appendix~\ref{apDeriv}), we obtain
that under the risk-neutral measure $\QM$ the price $U_t(T)$ is
\begin{equation} \label{bondPrice}
U_t(T, r, \hatr, y, z) = \m1_{\tau > t} f(t, T, r,\hatr, y, z, 0) +
\m1_{\tau \le t} f(t, T, r,\hatr, y, z, 1).
\end{equation}
Here the function $f(t, T, r,\hatr, y, z, 1) \equiv u(t, T, X), \ X = \{r,\hatr, y, z\}$ solves the PDE
\begin{equation} \label{PDE1}
\fp{u(t,T,X)}{t} + {\cal L} u(t,T,X) - r u(t,T,X) = 0,
\end{equation}
\noindent where the diffusion operator $\cal L$ reads
\begin{align} \label{Ldiff}
\cal L &= \frac{1}{2}\sigma_{r}^2 r\sop{}{r} + \frac{1}{2} \sigma_{\hatr}^2 \hatr \sop{u}{\hatr} + \frac{1}{2}\sigma_z^2 z^2 \sop{}{z} + \frac{1}{2}\sigma_y^2\sop{}{y}
+ \rho_{r \hatr} \sigma_r \sigma_{\hatr} \sqrt{r \hatr}\cp{}{r}{\hatr}  \\
&+ \rho_{rz}\sigma_r \sigma_z z\sqrt{r} \cp{}{r}{z}
+ \rho_{\hatr z} \sigma_{\hatr} \sigma_z z \sqrt{\hatr} \cp{}{z}{\hatr}
+ \rho_{ry}\sigma_r \sigma_y \sqrt{r} \cp{}{r}{y}
+ \rho_{\hatr y} \sigma_{\hatr} \sigma_y \sqrt{\hatr} \cp{}{y}{\hatr}
\nonumber \\
&+ \rho_{yz} \sigma_y \sigma_z z \cp{}{y}{z}
+ a(b-r)\fp{}{r}
+ \hat a(\hat b - \hatr) \fp{}{\hatr}
+ (r - \hatr) z \fp{}{z}
+ \kappa(\theta - y) \p{}{y}. \nonumber
\end{align}

The second function $f(t, T, r,\hatr, y, z, 0) \equiv v(t, T, X)$ solves the PDE
\begin{align} \label{PDE2}
\fp{v(t,T,X)}{t} &+ {\cal L} v(t,T,X) - r v(t,T,X)
- \lambda \gamma_z z \fp{v(t,T,X)}{z} \\
&+ \lambda \left[ u(t, T, X^+) -
v(t, T, X) \right] = 0, \qquad X^+ = \{r, \hatr(1+\gamma_\hatr), y, z(1+\gamma_z)\}. \nonumber
\end{align}
\noindent where according to \eqref{lambda}, $\lambda = e^y$.

The boundary conditions for this problem should be set at the boundaries of the unbounded domain $(r, \hatr, y, z) \in [0,\infty] \times [0,\infty] \times [-\infty,0] \times[0,\infty]$. However, this can be done in many different ways. As the value of the bond price is usually not known at the boundary, similarly to \cite{Brigo} we assume the second derivatives to vanish towards the boundaries
\begin{align} \label{bc}
& \sop{u}{\nu}\Big|_{\nu \uparrow 0} = \sop{u}{\nu}\Big|_{\nu \uparrow \infty} = 0, \quad \nu \in [r, \hatr], \\
&\sop{u}{y}\Big|_{y \uparrow 0} = \sop{u}{y}\Big|_{y \uparrow -\infty} = 0, \qquad
\sop{u}{z}\Big|_{z \uparrow 0} = \sop{u}{z}\Big|_{y \uparrow \infty} = 0. \nonumber
\end{align}

We assume that the default has not yet occured  at the validation time $t$, therefore, \eqref{bondPrice} reduces to
\begin{equation} \label{bondPrice1}
U_t(T, r, \hatr, y, z) = v(t, T, X).
\end{equation}
Therefore, it could be found by solving \eqref{PDE1}, \eqref{PDE2} as follows. Since the payoff in \eqref{payoff} is a sum of two terms, and our PDE is linear, it can be solved independently for each term. Then the solution is just a sum of the two.

\subsection{Solving the PDE for $w_t(T)$} \label{wtT}

The function $w_t(T)$ solves exactly the same set of PDEs as in \eqref{PDE1}, \eqref{PDE2} \footnote{The PDEs remain unchanged since the model is same, and only the continent claim $G(t,T,r,\hatr, y,z,d)$, which is a function of the same underlying processes, changes.}.
Therefore, it can be found in two steps.

\paragraph {Step 1}
We begin by solving the PDE in \eqref{PDE1} for function $u$. Since this function  corresponds to $d=1$, it describes the evolution of the bond price {\it at or after} default. Accordingly, the terminal condition for $u$ becomes $u(T, T, X) = 0$. Indeed, this payoff does not assume any recovery paid at default, therefore, the bond expires worthless. Then, a simple analysis shows that the function $u(t, T, X) \equiv 0$ is the solution at $d=1$ as it solves the equation itself and obeys the terminal and boundary conditions. Therefore, at this step the solution can be found analytically.

\paragraph {Step 2}  As the solution of the first step vanishes, it implies that $u(t, T, X^+) \equiv 0$ in \eqref{PDE2}.

By the definition before \eqref{PDE2}, the function $v$ corresponds to the states with no default. Accordingly, from \eqref{payoff} the payoff function (which is the terminal condition for \eqref{PDE2} at $t=T$) reads
\begin{equation} \label{tc2}
v(T,T,X) = z.
\end{equation}
The boundary conditions again are set as in \eqref{bc}.

The PDE \eqref{PDE2} for $v(t,T,X)$ now takes the form
\begin{align*}
\fp{v(t,T,X)}{t} &+ {\cal L} v(t,T,X) - (r + \lambda) v(t,T,X)
- \lambda \gamma_z z \fp{v(t,T,X)}{z} = 0,
\end{align*}
\noindent subject to the terminal condition $v(T,T,X) = z$. Then, obviously $w_t(T) = v(t,T,X)$.

It can be seen, that in case of no recovery, the defaultable bond price does depend on jump in the FX rate, but does not depend on the jump in the foreign interest rate.

\subsection{Solving the PDE for $g_t(\nu)$} \label{gtT}

As far as the second part of the payoff in \eqref{payoff} is concerned, it could be noticed that the integral in \eqref{payoff} is a Riemann--Stieltjes integral in $\nu$. Therefore, it can be approximated by a Riemann--Stieltjes sum where the continuous time interval $[t,T]$ could be replaced by a discrete uniform grid with sufficiently small step $\Delta \nu = h$. So
\begin{equation}\label{eq:Integral24}
\int_t^T g_t(\nu) d\nu \approx h \sum_{i=1}^N g_t(t_i),
\end{equation}
\noindent where $t_i = t + i h, \ i \in [0,N]$, $N = (T-t)/h$. Accordingly, each term in this sum can be computed independently by solving the corresponding pricing problem in \eqref{PDE1}, \eqref{PDE2} with the maturity $t_i$.

Note, that since the pricing problem in \eqref{PDE1}, \eqref{PDE2} is formulated via backward PDEs, computation of $g_t(t_i)$ for every maturity $t_i, \ i \in [1,m]$ requires an independent solution of such a problem. This could be significantly improved if instead of the backward PDE we would work with the forward one for the corresponding density function. In that case all $U_t(t_i), \ i \in [1,m]$ can be computed in one run (by a marching method). However, we leave this improvement to discuss in detail  elsewhere. Do not confuse $m$ and $N$ since $m$ is the total number of coupon payments, while $N$ is the number of discretisation steps in the integral \eqref{eq:Integral24}.

Again, it can be observed that the function $g_t(T)$ solves exactly the same set of PDEs as in \eqref{PDE1}, \eqref{PDE2}, and, thus, again it can be found in two steps.

\paragraph {Step 1}  The problem for $u$ should be solved subject to the terminal condition
\begin{equation} \label{tcG}
g_T(T) = z (1+\gamma_z).
\end{equation}
Indeed, by the definition of $g_t(T)$, we can set $t=T$ and condition on $R_t = r, \hatR_t = \hatr, Z_t = z, Y_t = y, d=1$. Then
\begin{align}
g_T(T)dT &= \mathbb{E}_t \left[ B(t,T) Z_T \m1_{\tau \in (T-dT,T]} \Big| t=T \right]
= z \mathbb{E}_t \left[\lambda_t dt | t=T \right]  = z e^y d T,
\end{align}
\noindent see \cite{Schonbucher2003}, Section~3.2.
However, the dynamics of $Z_t$ in \eqref{dzJump} implies that when the default occurs, the value of $Z_{\tau^-}$ jumps proportionally to the value $Z_\tau = Z_{\tau^-}(1+\gamma_z)$. Thus, we arrive at \eqref{tcG}.

\paragraph {Step 2} Having an explicit representation of the function $u(t, T, X)$ obtained as the solution of the previous step, one can find $u(t, T, X^+)$ as the values of parameters $\gamma_z, \gamma_\hatr$ are known, and the values of $\lambda$ are also given (for instance, at some grid which is used to numerically solve the PDE problem in Step~1). Then, \eqref{PDE2} can be solved with respect to $v(t, T, X)$.

By the definition before \eqref{PDE2}, the function $v$ corresponds to states with no defaults. Accordingly, the recovery is not paid, and the terminal condition for this step is $v(T,T,X) = 0$. This, however, does not mean that $v=0$ solves the problem. That is because \eqref{PDE2} contains the term $\lambda u(t, T, X^+) \ne 0$ (since the terminal condition at the previous step is not zero), and so $v \ne 0$ if $\lambda \ne 0$.

It can be seen that according to this structure in case of non-zero recovery the defaultable bond price does depend on jumps in both FX and foreign IR rates.

\section{From bond prices to CDS prices} \label{bond2cds}

As this paper is mostly dedicated to modeling Quanto CDS contracts, we use the setting developed in the previous sections for risky bonds and apply it to CDS contracts.
Let us remind that a CDS is a contract in which the protection buyer agrees to pay a periodic coupon to a protection seller in exchange for a potential cashflow in the
event of default of the CDS reference name before the maturity of the contract $T$.

We assume that a CDS contract is settled at time $t$ and assures protection to the CDS buyer until time $T$. We consider CDS coupons to be paid periodically with the payment time interval $\Delta t$, and there will be totally $m$ payments over the life of the contract, i.e., $m \Delta t = T-t$. Assuming unit notional, this implies the following expression for the CDS coupon leg $L_c$, \cite{LiptonSavescu2014, BrigoMorini2005}
\begin{equation}
L_c = \mathbb{E}_t\left[\sum_{i=1}^{m} c B(t,t_i)\Delta t \m1_{\tau > t_i}\right],
\end{equation}
\noindent where $c$ is the CDS coupon, $t_i$ is the payment date of the $i$-th coupon, and $B(t,t_i) = B_t/B_{t_i}$ is the stochastic discount factor.

However, if the default occurs in between of the predefined coupon payment dates, there must be an accrued amount from the nearest past payment date till the time of the default event $\tau$. The expected discounted accrued amount $L_a$ reads
\begin{equation}
L_a = \mathbb{E}_t\left[c  B(t,\tau) (\tau - t_{\beta(\tau)}) \m1_{t < t_\beta(\tau) \le \tau < T}\right],
\end{equation}
\noindent where $t_{\beta(\tau)}$ is the payment date preceding the default event. In other words, $\beta(\tau)$ is a piecewise constant function of the form
\[ \beta(\tau) = i, \quad \forall \tau: \ t_i < \tau < t_{i+1}. \]
These cashflows are paid by the contract buyer and received by the contract issuer. The opposite expected protection cashflow $L_p$ is
\begin{equation}
L_p = \mathbb{E}_t\left[(1 - \calR)B(t,\tau)\m1_{t < \tau \le T}\right],
\end{equation}
\noindent where the recovery rate $\calR$  is unknown beforehand, and is determined at or right after the default, e.g., in court. In modern mathematical finance theory it is customary to consider the recovery rate to be stochastic, see e.g., \cite{Cohen2017}) and references therein, however, throughout this paper we assume the recovery rate being constant and known in advance.

Further, we define the so-called \emph{premium} ${\cal L}_{pm} = L_c + L_a$ and \emph{protection} ${\cal L}_{pr} = L_p$ legs, and, as usual, define the CDS par spread $s$ as the coupon which equalizes these two legs and makes the CDS contract fair at time $t$. Similar to Section~\ref{zcbPrice}, if we price all instruments under the domestic money market measure $\QM$ we need to convert the payoffs to the domestic currency and discount by the domestic money market account. Then $s$ solves the equation
\begin{align} \label{eq:CDSequation}
\sum_{i=1}^{m} & \mathbb{E}_t \left[s Z_T B(t,t_i)\Delta t \m1_{\tau > t_i}\right]
+ \mathbb{E}_t\left[s Z_T B(t,\tau)(\tau - t_{\beta(\tau)})\m1_{t<\tau<T}\right]  \\
&= \mathbb{E}_t\left[(1 - \calR)Z_\tau B(t,\tau)\m1_{t<\tau\leq T}\right]. \nonumber
\end{align}

In the spirit of \cite{ES2006} and \cite{BrigoSlide}, we develop a numerical procedure for finding the par spread $s$ from the bond prices. Consider each term in \eqref{eq:CDSequation} separately.

\paragraph{Coupons} For the coupon payment one has
\begin{align} \label{eqCoupon}
L_c &= \mathbb{E}_t \left[ \sum_{i=1}^{m} s Z_{t_i} B(t,t_i) \Delta t \m1_{\tau \ge t_i} \right] =  s\Delta t \sum_{i=1}^m \mathbb{E}_t \left[ Z_{t_i} B(t,t_i) \m1_{\tau \ge t_i} \right] = s \Delta t \sum_{i=1}^m w_t(t_i).
\end{align}
\noindent where $t_m = T$. Computation of $w_t(T)$ is described in Section~\ref{wtT}.

Note, that as follows from the analysis of the previous section, $w_t(T)$ (and, respectively, the coupon payments)  does depend on the jump in the FX rate, but does not depend on the jumps in the foreign interest rate which is financially reasonable.

\paragraph{Protection leg} A similar approach is provided for the protection leg
\begin{align} \label{eqProtection}
L_p &= \mathbb{E}_t \left[(1-\calR) Z_{\tau} B(t,\tau)\m1_{t < \tau \le T} \right]
=  (1-\calR) \int_{t}^{T} \mathbb{E}_t \left[Z_\nu B(t,\nu) \m1_{\tau \in (\nu-d\nu,\nu]} \right] d\nu \\
&=  (1-\calR)\int_{t}^{T} g_t(\nu) d \nu, \nonumber
\end{align}
\noindent where computation of $g_t(T)$ is described in Section~\ref{gtT}.

\paragraph{Accrued payments} For the accrued payment one has
\begin{align} \label{eqAccrued}
L_a &=  \mathbb{E}_t \left[ s Z_\tau B(t,\tau) (\tau - t_{\beta(\tau)}) \frac{\m1_{t < \tau < T}}{d\nu} \right]
= s \int_t^T \mathbb{E}_t \left[ Z_\nu B(t,\nu) (\nu - t_{\beta(\nu)}) \frac{\m1_{\tau \in (\nu-d\nu,\nu]}}{d\nu}  \right] d\nu \\
&= s \sum_{i=0}^{m-1} \Big\{\int_{t_i}^{t_{i+1}} (\nu - t_i) \mathbb{E}_t \left[ Z_\nu B(t,\nu) \frac{\m1_{\tau \in (\nu-d\nu,\nu]}}{d\nu}  \right] d\nu \Big\}
= s \sum_{i=0}^{m-1} \int_{t_i}^{t_{i+1}} (\nu - t_i)g_t(\nu) d\nu, \nonumber
\end{align}
\noindent where $t_0 \equiv t$, and $t_m \equiv T$.

As was mentioned in Section~\ref{gtT}, both final integrals in \eqref{eqProtection}, \eqref{eqAccrued} are Riemann--Stieltjes integrals in $\nu$. Therefore, each one can be approximated by a Riemann--Stieltjes sum where the continuous time interval $[t,T]$ could be replaced by a discrete uniform grid with a sufficiently small step $\Delta \nu = h$.

Now we have all necessary componets to compute the CDS spread. Introducing new notation
\begin{align} \label{approx1}
A_i &= \int_{t_i}^{t_{i+1}} w_t(\nu) d\nu \approx h \sum_{k=1}^N w_t(\nu_k), \\
B_i &= \int_{t_i}^{t_{i+1}} g_t(\nu) d\nu \approx h \sum_{k=1}^N g_t(\nu_k) \nonumber \\
C_i &= \int_{t_i}^{t_{i+1}} \nu g_t(\nu) d\nu \approx h \sum_{k=1}^N \nu_k g_t(\nu_k) \nonumber \\
\nu_k &= t_i + k h, \quad k=1,\ldots,N, \quad h = (t_{i+1} - t_i)/N, \nonumber
\end{align}
\noindent we re-write \eqref{eqCoupon}, \eqref{eqProtection} and \eqref{eqAccrued} in the form
\begin{align} \label{approx2}
L_p &= (1-\calR) \sum_{i=1}^{m} B_i, \qquad
L_c = s \Delta t \sum_{i=1}^m A_i, \qquad
L_a = s \sum_{i=1}^{m} \left[ C_i - t_i B_i \right].
\end{align}
Finally, combining together \eqref{eq:CDSequation} and \eqref{approx2} we obtain
\begin{align} \label{eqParSpread}
s = (1-\calR) \dfrac{\sum_{i=1}^{m} B_i}{\sum_{i=1}^{m} \left[ \Delta t A_i + C_i - t_i B_i\right]}.
\end{align}

%%%%%%%%%%%%%%%%%%%%%%%%%%%%%%%%%%%%%%%%%%%%%%%%%
\section{Radial Basis Function Partition of Unity Method} \label{numMethod}
%%%%%%%%%%%%%%%%%%%%%%%%%%%%%%%%%%%%%%%%%%%%%%%%%
In order to numerically solve \eqref{PDE1}, \eqref{PDE2} subject to the corresponding terminal and boundary conditions we use a radial basis function method. Radial basis function methods become increasingly popular for applications in computational finance, e.g.,  \cite{YCHon3,Fasshauer2,Pettersson}, thanks to their high order convergence that allows for obtaining a high resolution scheme using just a few discretization nodes. This is a crucial property when solving various multi-dimensional problems, e.g., pricing derivatives written on several assets (basket options), or those for models whose settings use several stochastic factors. Indeed, all these models suffer immensely from the curse of dimensionality, in particular, an increasing storage (memory) becomes the dominant limiting factor. This, however, can be successfully overcome by using the RBF methods. For instance, in \cite{Shcherbakov} it is shown that standard finite difference methods require about three times as many computational nodes per dimension as RBF methods to obtain the same accuracy, thus, significantly reducing the memory consumption.

Nevertheless, it should be emphasized, that the original global RBF method is computationally very expensive and rather unstable due to dense and ill-conditioned coefficient matrices\footnote{More details could be found, e.g., in \cite{Fasshauer}.}. This is a consequence of the global connections between the basis functions. Therefore, here we eliminate from the global RBF method in favour of its localised version based on the idea of partition of unity. The partition of unity method was originally introduced by \cite{Melenk} for finite element methods, and later adapted for the RBF methods by several authors, \cite{Safdari,Shcherbakov}. This approach (which further on is referred as RBF-PUM) enables a significant reduction in the number of non-zero elements that remain in the coefficient matrix, hence, lowering the computational intensity required for solving the system. In addition, this concept is supported, say in Matlab, by making use of sparse operations. Typically, as applied to our problem of pricing Quanto CDS, only about one percent of all elements remain to have non-zero values.

In order to construct an RBF-PUM approximation we start by defining an open cover $\{\Omega_j\}_{j=1}^{P}$ of our computational domain $\Omega$ such that
\begin{equation}
  \Omega \subseteq \bigcup_{j=1}^{P} \Omega_j.
\end{equation}
We select the patches $\Omega_j$ to be of a spherical form. Inside each patch a local RBF approximation of the solution $u$ is defined as
\begin{equation}\label{localRBF}
  \tilde u_j(x)= \sum_{i=1}^{n_j}\lambda_i^j \phi(\varepsilon, || x - x_i^j ||),
\end{equation}
\noindent where $n_j$ is the number of computational nodes belonging to the patch $\Omega_j$, $\phi(\varepsilon, ||x - x_i^j ||)$ is the $i$-th basis function centred at $x_i^j$, which is the $i$-th local node in the $j$-th patch $\Omega_j$, $\varepsilon$ is the shape parameter that determines the widths of basis functions, and $\lambda_i^{j}$ are the unknown coefficients. Some popular choices of the basis functions are listed in Table~\ref{TabRBF}, while their behavior as a function of the parameter $\varepsilon$ is presented in Fig.~\ref{BasisFunc}.

\begin{table}[H]
\begin{center}
\begin{tabular}{ l  c  c  c  r  }
%\hline\hline
RBF & & &  & $\phi(\varepsilon, r)$   \\
\hline
 Gaussian (GA) &  & &  &  $\exp{(-\varepsilon^2r^2)}$ \\
 Multiquadric (MQ) &  & &  & $\sqrt{1+\varepsilon^2r^2}$ \\
 Inverse Multiquadric (IMQ) & & &  & $1/\sqrt{1+\varepsilon^2r^2}$ \\
 Inverse Quadratic (IQ) & & &  & $1/(1+\varepsilon^2r^2)$ \\
\hline
%\hline\hline
\end{tabular}
\caption{Commonly used radial basis functions.}
\label{TabRBF}
\end{center}
\end{table}
\begin{figure}[H]
\centering
\subfigure[Gaussian]{\label{fig:a}\includegraphics[width=0.32\textwidth]{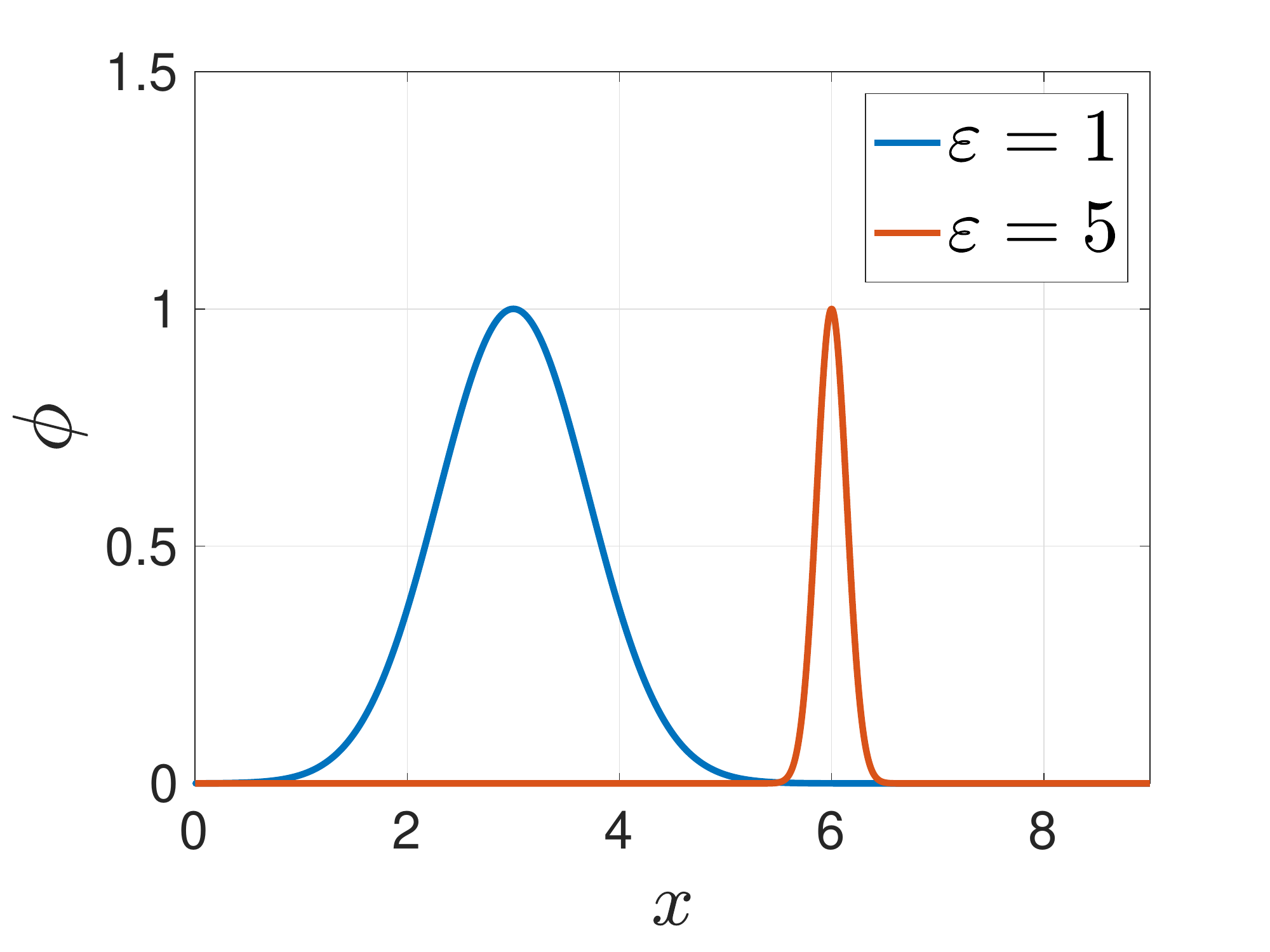}}
\subfigure[Multiquadric]{\label{fig:b}\includegraphics[width=0.32\textwidth]{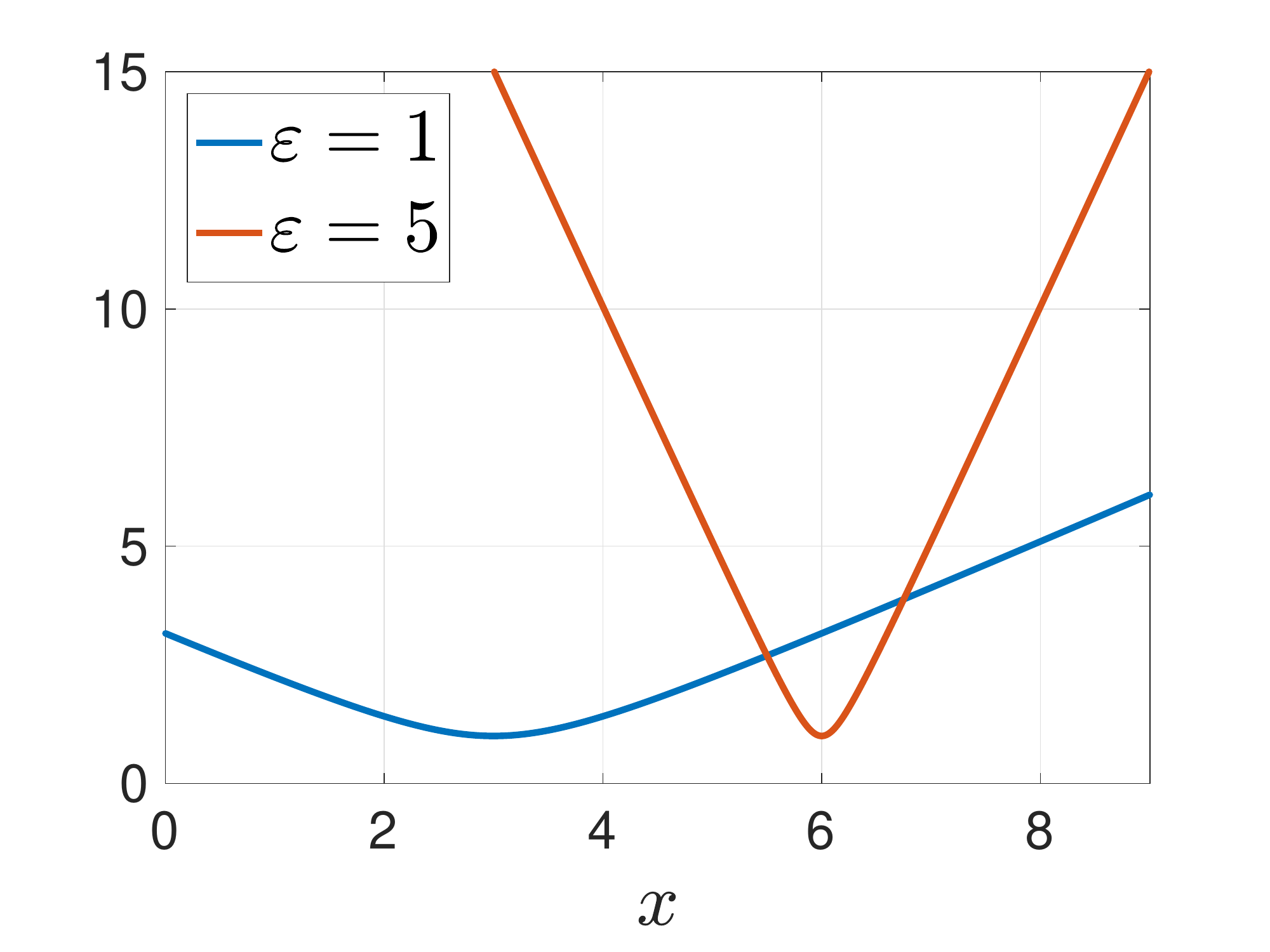}}
\subfigure[Inverse quadratic]{\label{fig:c}\includegraphics[width=0.32\textwidth]{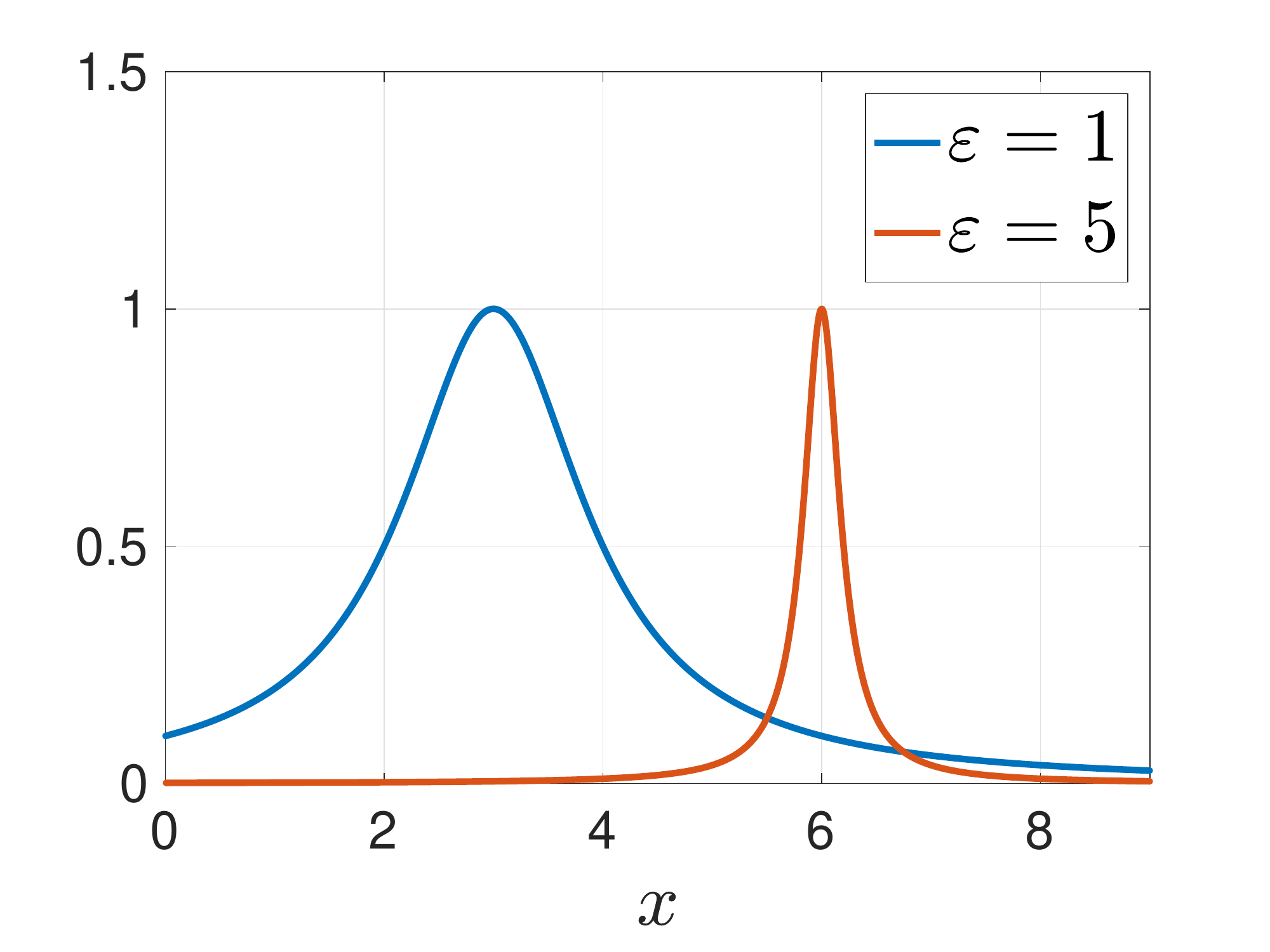}}
\caption{Commonly used basis functions with respect to the value of the shape  parameter~$\varepsilon$.}
\label{BasisFunc}
\end{figure}

In addition to the patches, we also construct partition of unity weight functions
$w_j(x),\, j = 1, \ldots, P$, subordinated to the open cover, such that
\begin{equation}
  \sum_{j=1}^{P} w_j(x) = 1, \quad \forall x\in \Omega.
\end{equation}
Functions $w_j(x)$ can be obtained, e.g., by Shepard's method, \cite{Shepard}, from compactly supported generating functions $\varphi_j(x)$
\begin{equation}
  w_j(x) = \frac{\varphi_j(x)}{\sum_{i=1}^{P} \varphi_i(x)}, \quad j=1,\ldots,P, \quad \forall x\in \Omega.
\end{equation}
The generation functions $\varphi_j(x)$ must fulfil some smoothness requirements. For instance, for the problem considered in this paper
they should be at least $C^{2}(\Omega)$. To proceed, as a suitable candidate for $\varphi_j(x)$ we choose fifth-order Wendland's functions, \cite{Wendland}
\begin{equation}
  \varphi(r) = (5r+1)(1-r)^5_{+}, \quad r \in \mathbb{R},
\end{equation}
\noindent with the support $\varphi(r) \in \mathbb{B}^{4}(0, 1)$, where $\mathbb{B}^4(0,1)$ is a unit four-dimensional ball centred at the origin. In order to map the generating function to the patch $\Omega_j$ with the centre $c_j$ and radius $\rho_j$, it is shifted and scaled as
\begin{equation}
 \varphi_{j}(x) = \varphi_j \left( \frac{||x- c_j||}{\rho_j}\right), \quad \forall x \in \Omega.
\end{equation}
Further we blend the local RBF approximations with the partitions of unity weight and obtain a combined RBF-PUM solution $\tilde u(x)$ as
\begin{equation}\label{RBFPUMapprox}
 \tilde u(x) = \sum_{j=1}^{P}w_j(x) \tilde u_j(x).
\end{equation}
The RBF-PUM approximation in the given form allows to maintain accuracy similar to that of the global method while significantly reducing the computational effort (see e.g., \cite{Shcherbakov}, \cite{Ahlkrona}). Moreover, it was shown in \cite{vonSydow} that RBF-PUM is the most efficient numerical method for higher-dimensional problems among deterministic methods that rely on a node discretization.

%%%%%%%%%%%%%%%%%%%%%%%%%%%%%%%%%%%%%%%%%%%%%%%%%
\section{Numerical Experiments} \label{experiments}
%%%%%%%%%%%%%%%%%%%%%%%%%%%%%%%%%%%%%%%%%%%%%%%%%

In this section we perform numerical experiments to find the Quanto-adjusted CDS par spread value $s$ and its sensitivity to market conditions. The par spread is computed as in \eqref{eqParSpread} while the bond price is obtained from \eqref{bondPrice1} by approximating the PDEs in \eqref{PDE1}, \eqref{PDE2} using radial basis function partition of unity method with $1296$ patches.  We select Gaussian functions to construct a finite RBF basis on $28561$ nodes. As $[r,\hatr, z] \in [0,\infty)$ and $y \in (-\infty,\infty)$, we truncate each semi-infinite ot infinite domain of definition sufficiently far away from the evaluation point, so an error brought by this truncation is relatively small. In particular, we use $r_{\min} = \hat r_{\min} = z_{\min} = 0$, $y_{\min} = -6$, $r_{\max} = \hat r_{\min} = z_{\min} = 4$, $y_{\max} = -2$. Accordingly, we move the boundary conditions, defined in \eqref{bc}, to the boundaries of this truncated domain.

Note, that in our numerical method (see Section~\ref{numMethod}), we substitute \eqref{bc} into the pricing PDEs \eqref{PDE1}, \eqref{PDE2} and then derive a corresponding reduced form discrete (boundary) operator. As this explicitly incorporates the boundary conditions into the pricing scheme, the latter can be implemented uniformly with no extra check that the boundary conditions are satisfied\footnote{Our experience shows that this approach works better and provides a more stable RBF approximation.}.

For marching in time we use the backward differentiation formula of second order (BDF-2), \cite{BDFbook}. In order to compute the accrued amount $L_a$ as in \eqref{approx2} we use the time discretisation with two-weeks intervals. The method is implemented in Matlab 2017a, and the experiments were run on a MacBook Pro with a Core i7 processor with 16 GB RAM.
%The average execution time of the program was 70 seconds.

To investigate Quanto effects and their impact on the price of a CDS contract, we consider two similar CDS contracts. The first one is traded in the foreign economy, e.g., in Italy, but is priced under the domestic risk-neutral $\QM$-measure, hence is denominated in the domestic currency (US dollars). To find the price of this contract our approach described in the previous sections is utilized. The second CDS is the same contract which is traded in the domestic economy and is also priced in the domestic currency. As such, its price can be obtained by solving the same problem as for the first CDS, but where the equations for the foreign interest rate $\hatR_t$ and the FX rate $Z_t$ are excluded from consideration. Accordingly, all related correlations which include index $z$ and $\hatr$ vanish, and the no-jumps framework is used. However, the terminal conditions remain the same as in Section~\ref{bond2cds} as they are already expressed in the domestic currency\footnote{Alternatively, the whole four-dimensional framework could be used if one sets $z=1, \hatr = r, \gamma_z = \hat a = \sigma_\hatr = \gamma_\hatr = 0$, and $\rho_{\cdot,z} = \rho_{\cdot,\hatr} = \rho_{z,\hatr} = 0$, where $\langle \cdot \rangle \in [r, z, y]$.}.

Below we denote the CDS spread found by using the first contract as $s$, and the second one as~$s_d$. So the impact of Quanto effects could be determined as the difference between these two spreads
\begin{equation}
\Delta s = s - s_d,
\end{equation}
\noindent which below is quoted as ``basis" spread.

A default set of parameter values used in our numerical experiments is given in Table~\ref{TabParam}. It is also assumed that in this default set all correlations are zero. If not stated otherwise, we use these values and assume the absence of jumps in the FX and foreign interest rates. The reference 5Y CDS par spread value $s_d$ under these assumptions is $s_d = 365$ bps.
\begin{table}[H]
\begin{center}
\begin{tabular}{c c c c c c c c}
%\hline\hline
\multicolumn{8}{ c }{Interest rates} \\
\hline
$r$ & $a$ & $b$ & $\sigma_r$ & $\hat r$ & $\hat a$ & $\hat b$ & $\sigma_{\hat{r}}$ \\
0.02 & 0.08 & 0.1 & 0.01 & 0.03 & 0.08 & 0.1 & 0.08 \\
\hline
&&&&&&& \\
\multicolumn{8}{ c }{Hazard and FX rates, Tenor, and Recovery} \\
\hline
$y$ & $a_y$ & $b_y$ & $\sigma_y$ & $z$ & $\sigma_z$ & $T$ & $\mathcal{R}$ \\
-4.089 & 0.0001 & -210 & 0.4 & 1.15 & 0.1 & 5 & 0.45\\
\hline
%\hline\hline
\end{tabular}
\caption{The default set of parameter values used in the experiments.}
\label{TabParam}
\end{center}
\end{table}

The impact of the jump amplitude on the basis spread is presented in Fig.~\ref{fig:Gammas} for jumps in the interest rate (left panel) and exchange rate (right panel). In the absence of jumps ($\gamma_\hatr = 0$ or $\gamma_z = 0$) the domestic and foreign spreads have a basis about 3~bps. This is close to the normal situation where no currency and interest rate depreciation occurs. In fact, this was the case until recently when Quanto effects were not taken into account. For example, Greek CDS with payments in dollars and in euros were traded with a 1 bp difference in 2006, \cite{IFR2011}.
\begin{figure}[H]
\centering
\subfigure{\includegraphics[width=0.4\textwidth]{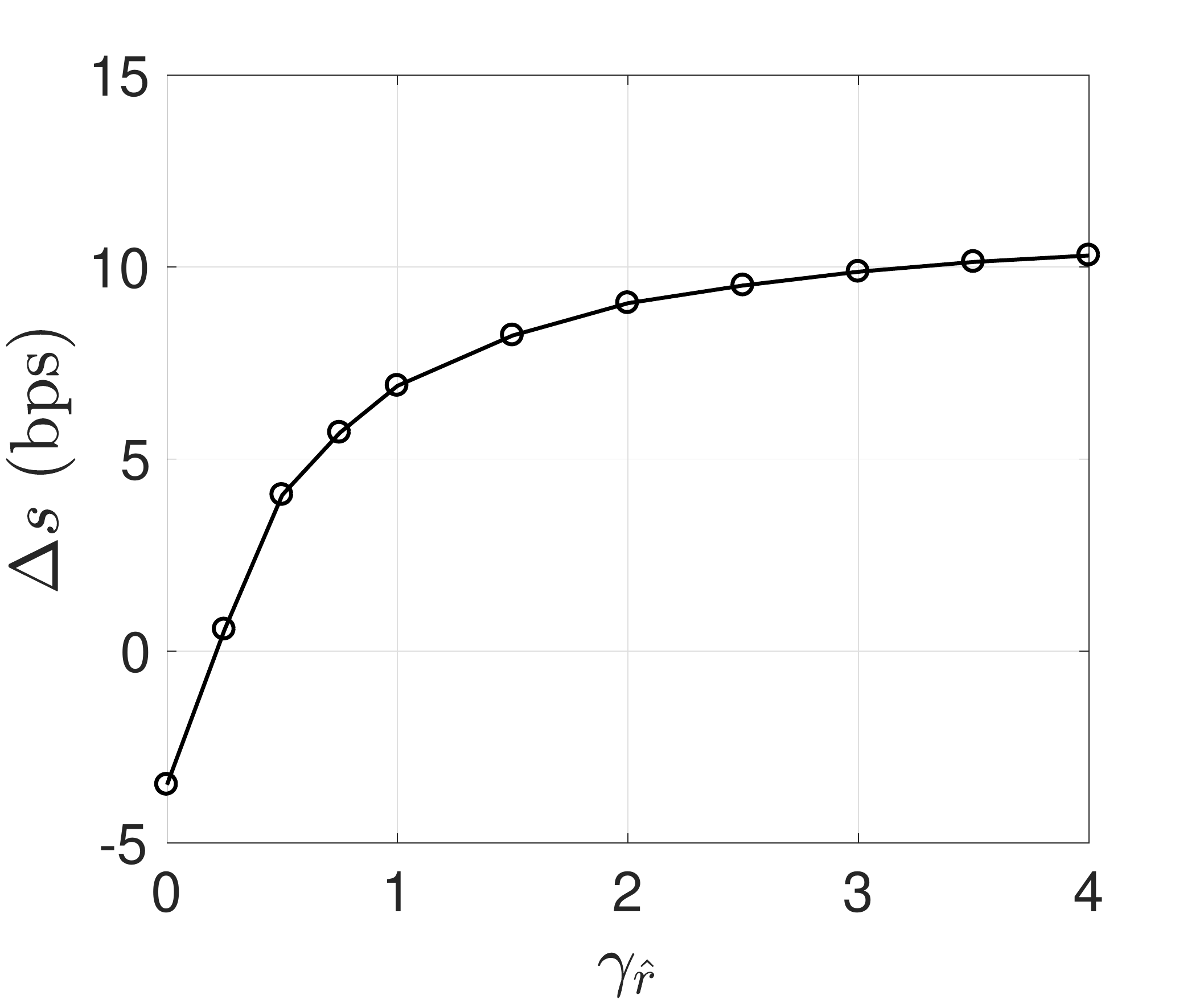}}
\hspace{0.7cm}
\subfigure{\includegraphics[width=0.4\textwidth]{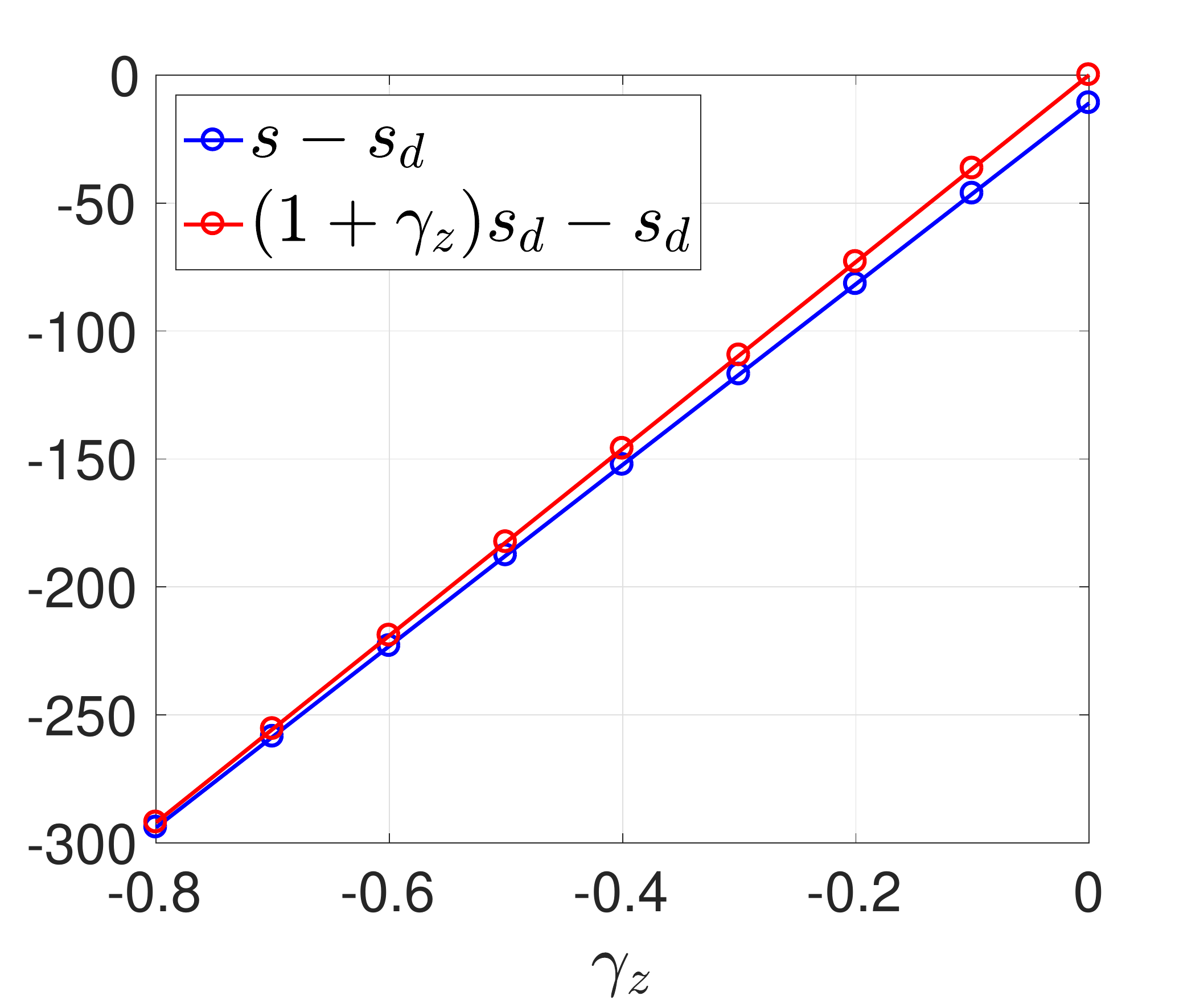}}
\caption{The influence of the jump-at-default amplitude on the 5Y CDS par spread.}
\label{fig:Gammas}
\end{figure}

The results displayed in the left panel of Fig.~\ref{fig:Gammas} demonstrate that the impact of jump in $\hat{R}_t$ increases  rapidly for $\gamma_{\hatr}\in[0,2]$ and then saturates at some level. We explain this saturation by investor's indifference to whether the interest rate increases by $300\%$ or $400\%$ since the interest rate level does not directly affect the protection amount, rather it influences the investment climate in the foreign economy. In contrast, the FX rate has an immediate impact (right panel) on the protection since a depreciation of the foreign currency diminishes the amount being paid out when converted to the US dollars. Through the well-known approximation of the hazard rate via the spread and bond recovery rate\footnote{Which is correct if the hazard rate $\lambda_t$ is constant.}
\[ \lambda\approx \frac{s}{1-\mathcal{R}}, \]
\noindent and using the results in~\cite{Brigo}, we identify that
\[ s \approx (1+\gamma_z)s_d. \]
That is, the CDS spread in the foreign currency is approximately proportional to the reference USD spread with the coefficient $(1+\gamma_z)$. Therefore, in the case of the foreign currency devaluation the coupon payments in the foreign currency should be lower. It can be observe that the results provided by our model perfectly align with this intuition.

We emphasize, that since the effect of the jump-at-default in the FX rate was thoroughly investigated in \cite{Brigo}\footnote{In \cite{Brigo}, however, only constant foreign and domestic interest rates are considered, while in this paper they are stochastic even in the no-jumps framework.}, in this paper we mainly focus on examining the impact of the jump-at-default in the foreign interest rate. However, influence of the other model parameters is also investigated and reported.
\begin{figure}[H]
\centering
\includegraphics[width=0.7\textwidth]{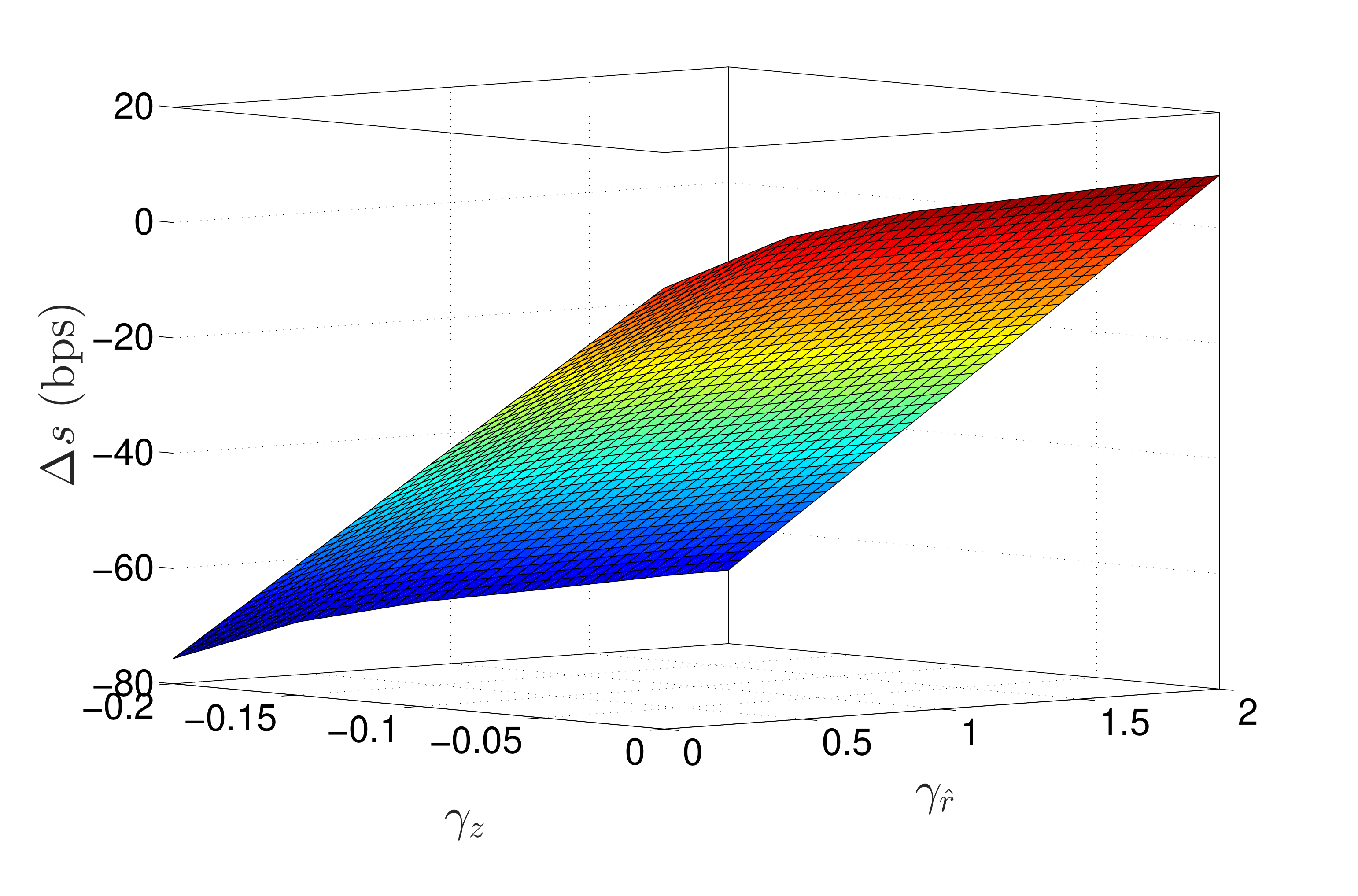}
\caption{Basis spread as a function of the jump amplitude in the foreign exchange and interest rates.}
\label{fig:SurfGammas}
\end{figure}

In Fig.~\ref{fig:SurfGammas} the joint influence of jumps in the FX and foreign IR on the value of the basis spread is presented. It can be seen that the jump-at-default in $\hatR_t$, which occurs simultaneously with the jump-at-default in $Z_t$, decreases the basis spread magnitude as compared with a similar case where $\hatR_t$  does not jump. This decrease slightly depends on the level of $\gamma_z$ and for our set of parameters is about 10 bps. To better illustrate this point Fig.~\ref{fig:GamR_GamZ} represents some slices of the surface in Fig.~\ref{fig:SurfGammas}. It can be seen that the smaller is $\gamma_z$ the bigger is the impact of $\gamma_\hatr$, which, however, reaches some saturation at $\gamma_\hatr \approx 4$.
\begin{figure}[H]
\centering
\includegraphics[width=0.7\textwidth]{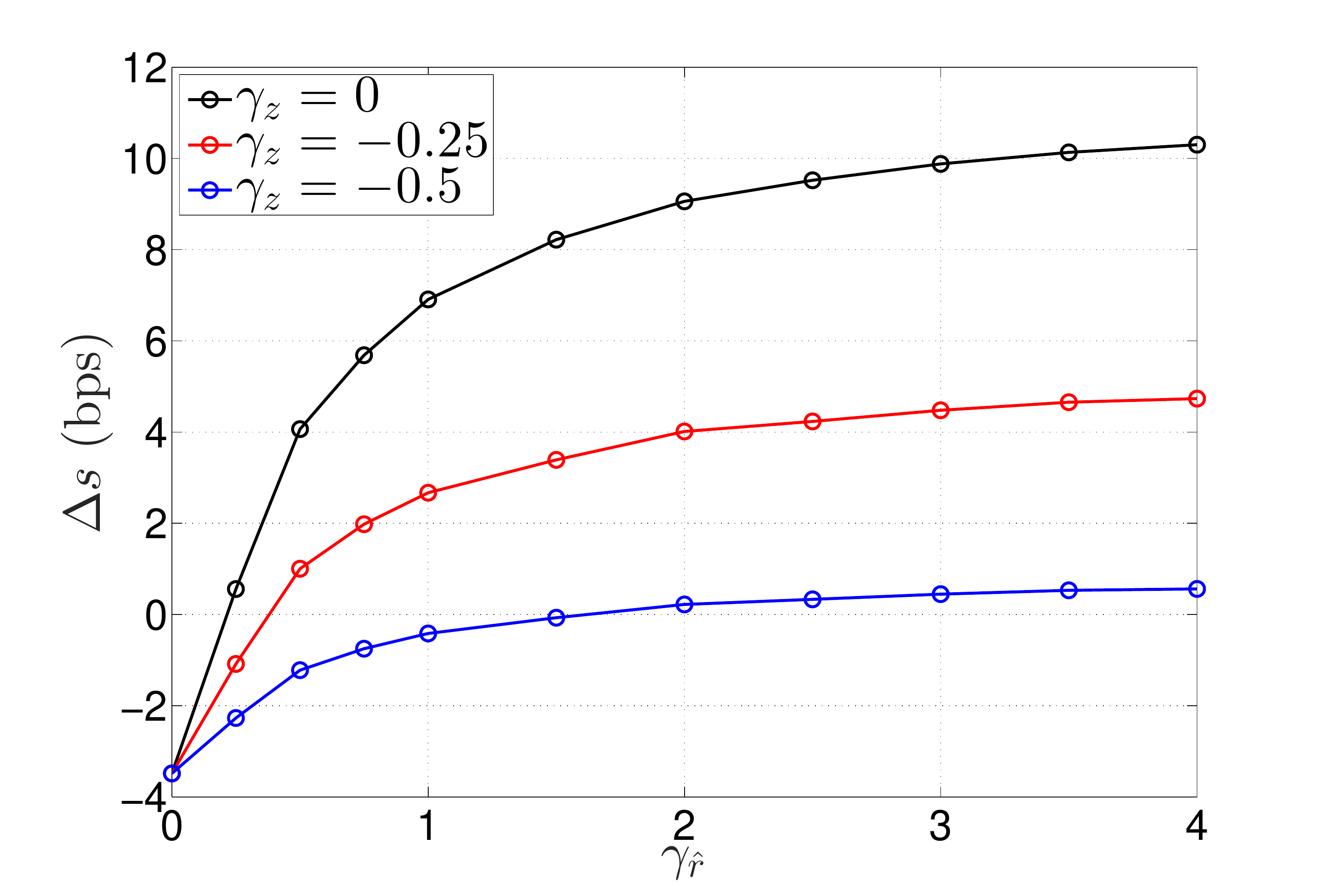}
\caption{The influence of the jump amplitude $\gamma_\hatr$ at various values of the jump amplitude $\gamma_z$. Note, the lines are shifted to start from the same point.}
\label{fig:GamR_GamZ}
\end{figure}

In the next series of experiments we look at the influence of correlations among the stochastic factors on the Quanto-adjusted CDS value. The results presented in Fig.~\ref{fig:Correlations} indicate that only the correlations between the hazard rate $\lambda_t$ (or $Y_t$) and the stochastic factors that experience a jump-at-default, $\hatR_t, Z_t$, are relevant. The impact of the correlation between the hazard and FX rates $\rho_{yz}$ can range in 45 bps, while the impact of the correlation $\rho_{y\hatr}$ between the hazard rate and the foreign interest rate does not exceed 3 bps.
\begin{figure}[H]
\centering
    \subfigure{\includegraphics[width=0.32\textwidth]{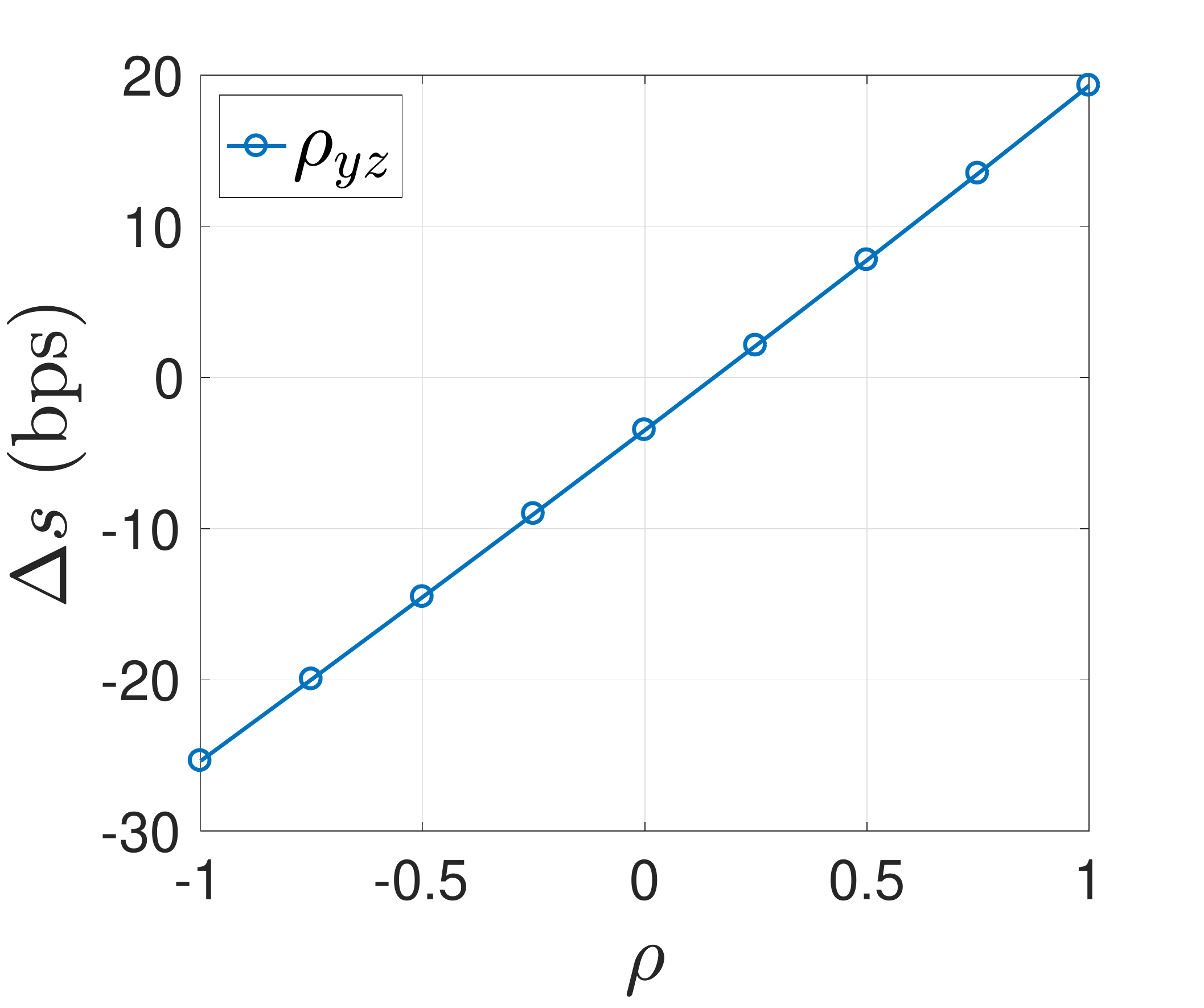}}
     \subfigure{\includegraphics[width=0.32\textwidth]{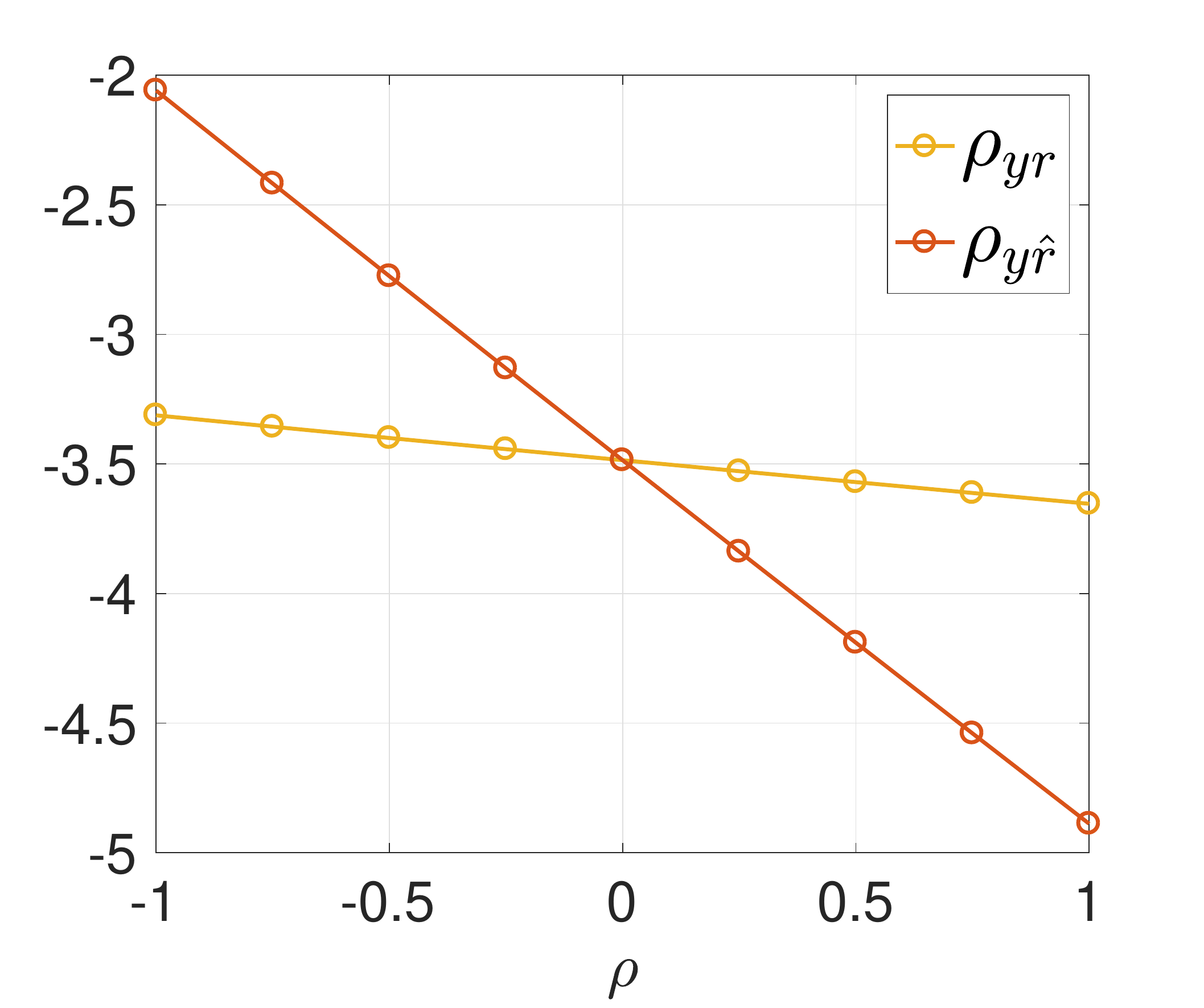}}
     \subfigure{\includegraphics[width=0.32\textwidth]{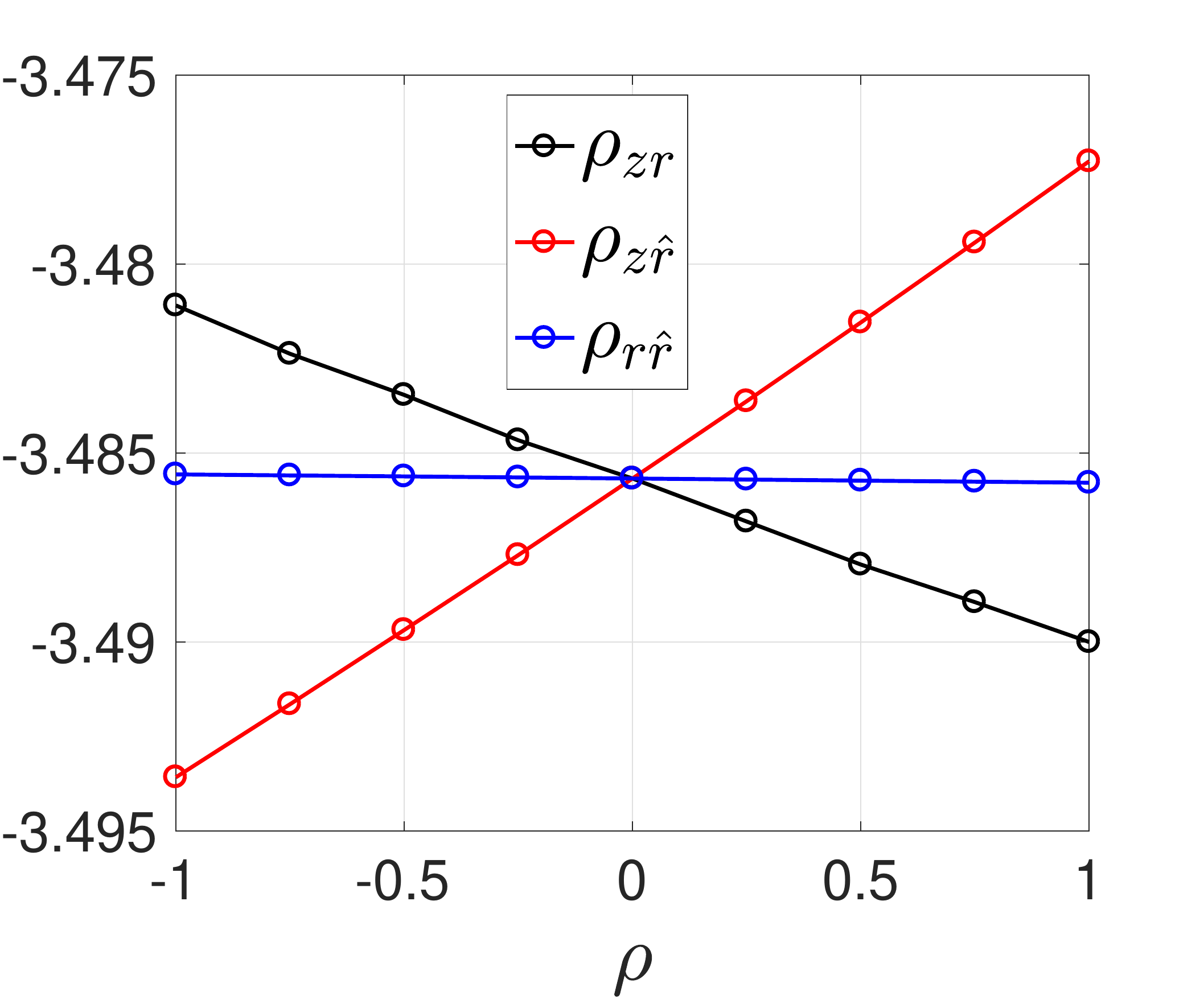}}
\caption{The influence of correlations on the 5Y  CDS par spread. No jumps-at-default are assumed.}
\label{fig:Correlations}
\end{figure}

Fig.~\ref{fig:GamRcorrs} shows how the level of correlation between the foreign interest rate $\hatR_t$ and the other three stochastic factors affects the basis spread
at various values of $\gamma_\hatr$ at $\gamma_z = 0$. In accordance with what was already mentioned, the results show that the correlations just slightly affect the basis spread value, except correlations with the hazard rate $\rho_{yz}, \rho_{y\hatr}$.
\begin{figure}[H]
\centering
     \subfigure{\includegraphics[width=0.32\textwidth]{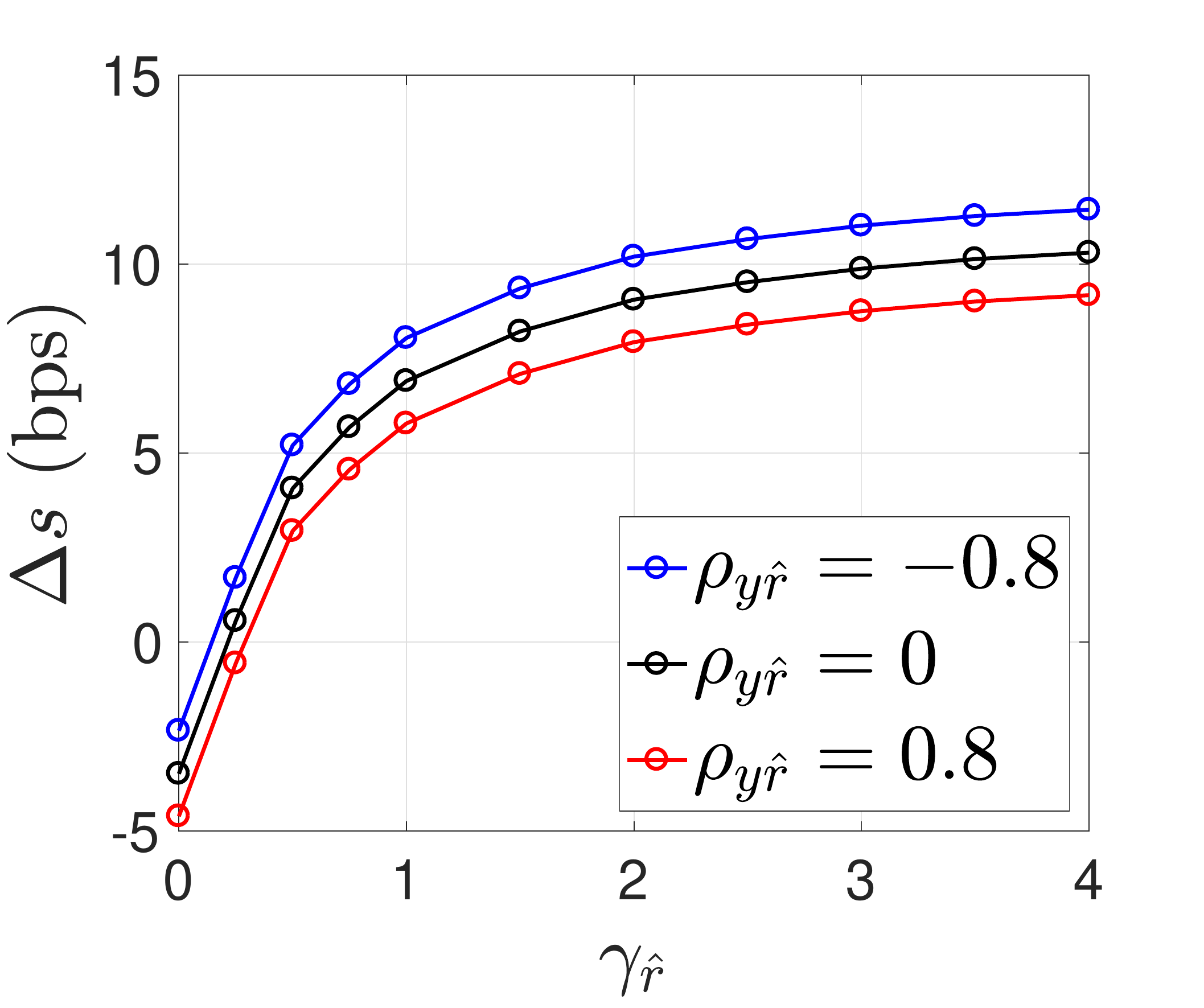}}
     \subfigure{\includegraphics[width=0.32\textwidth]{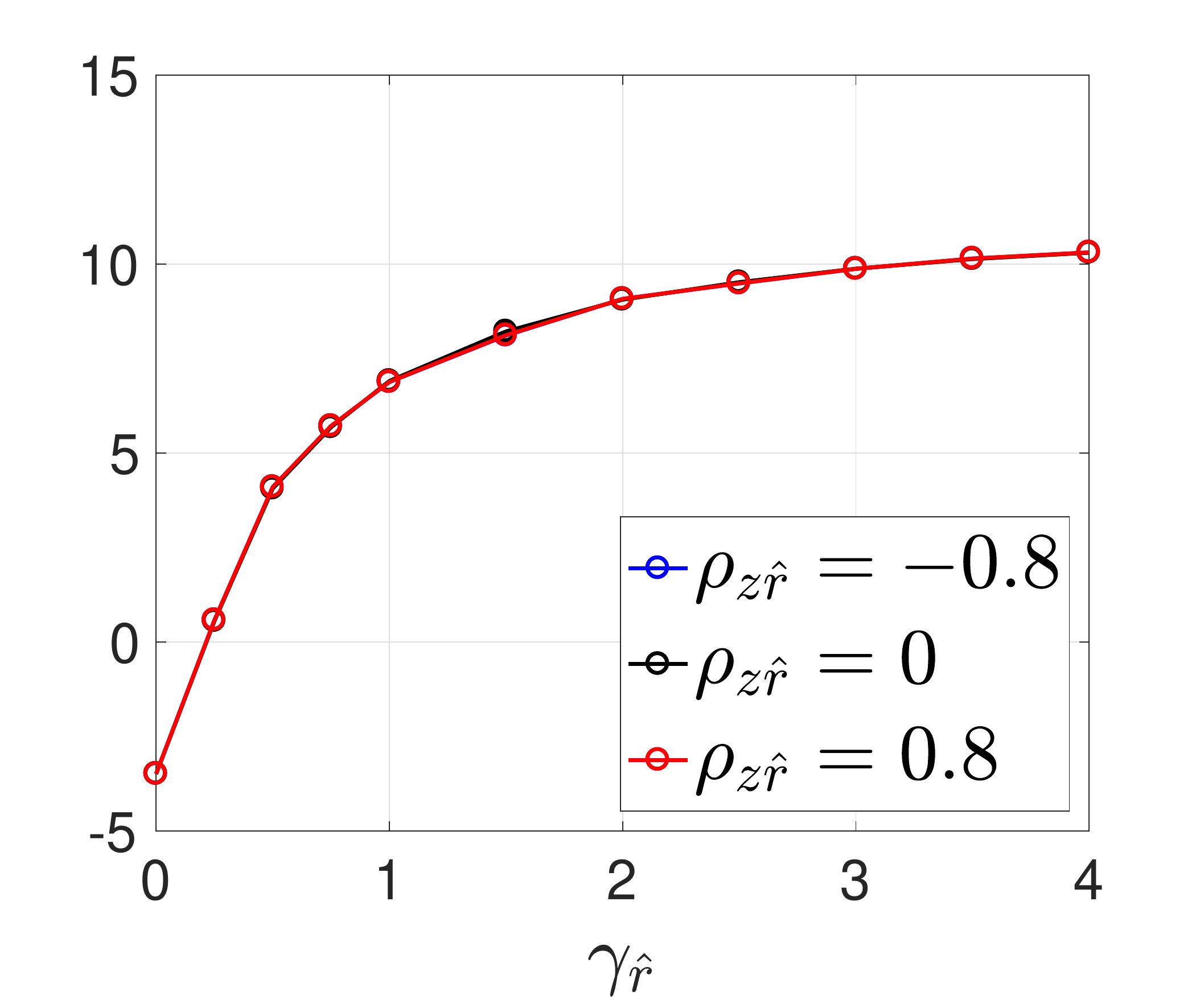}}
     \subfigure{\includegraphics[width=0.32\textwidth]{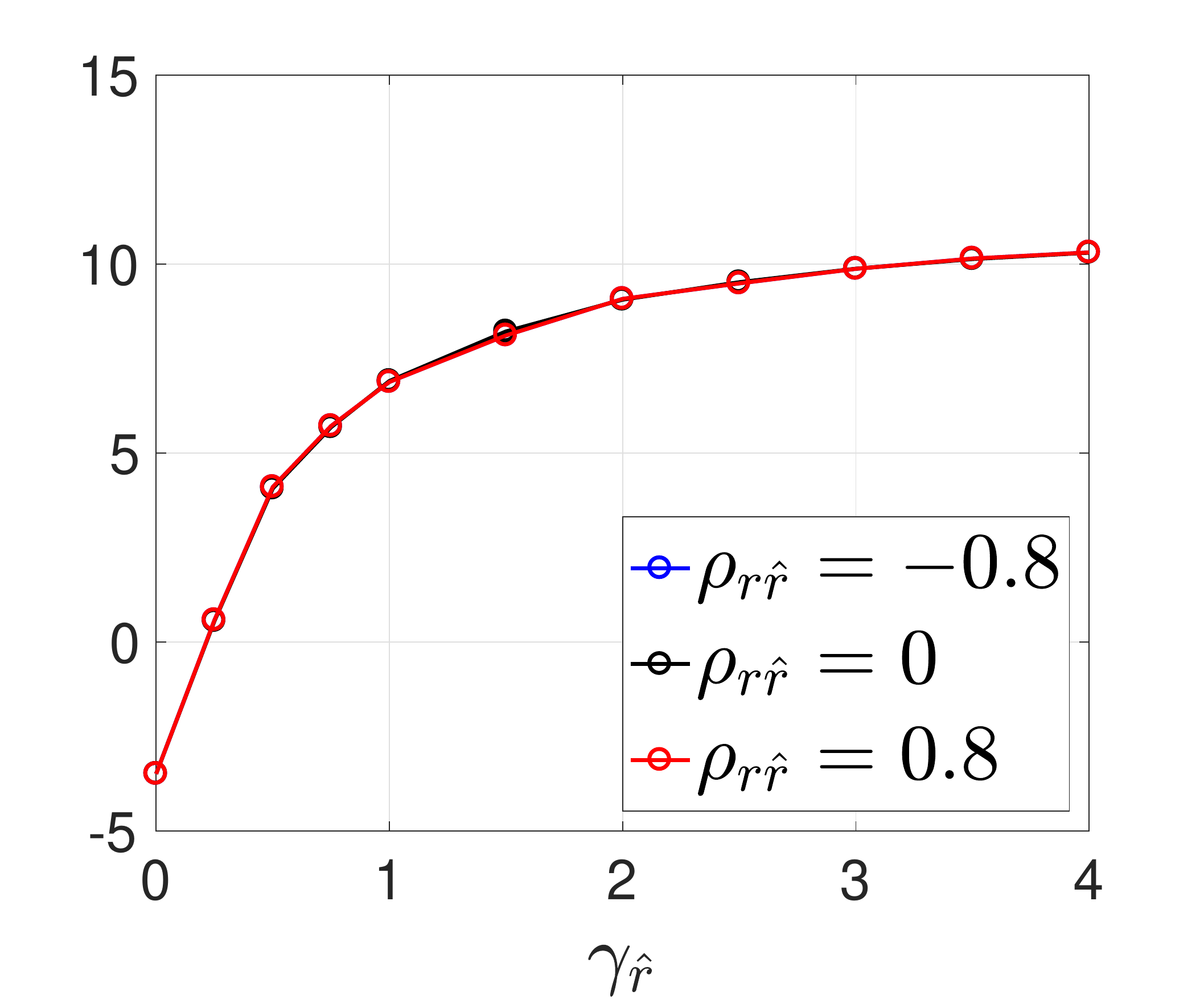}}
\caption{The influence of correlations on the basis spread for various $\gamma_\hatr$ at $\gamma_z = 0$.}
\label{fig:GamRcorrs}
\end{figure}

Fig.~\ref{fig:Volatility} shows the sensitivity of the foreign CDS to volatilities of the stochastic factors. We notice that the impact of the hazard rate volatility $\sigma_y$ is the strongest, and under the jump-free setup can make the CDS quotes varying in range in 17 bps. The effect of the FX rate volatility $\sigma_z$ is slightly weaker, while the effect of the interest rate volatilities $\sigma_r, \sigma_\hatr$ is almost negligible.
\begin{figure}[H]
\centering
     \subfigure{\includegraphics[width=0.32\textwidth]{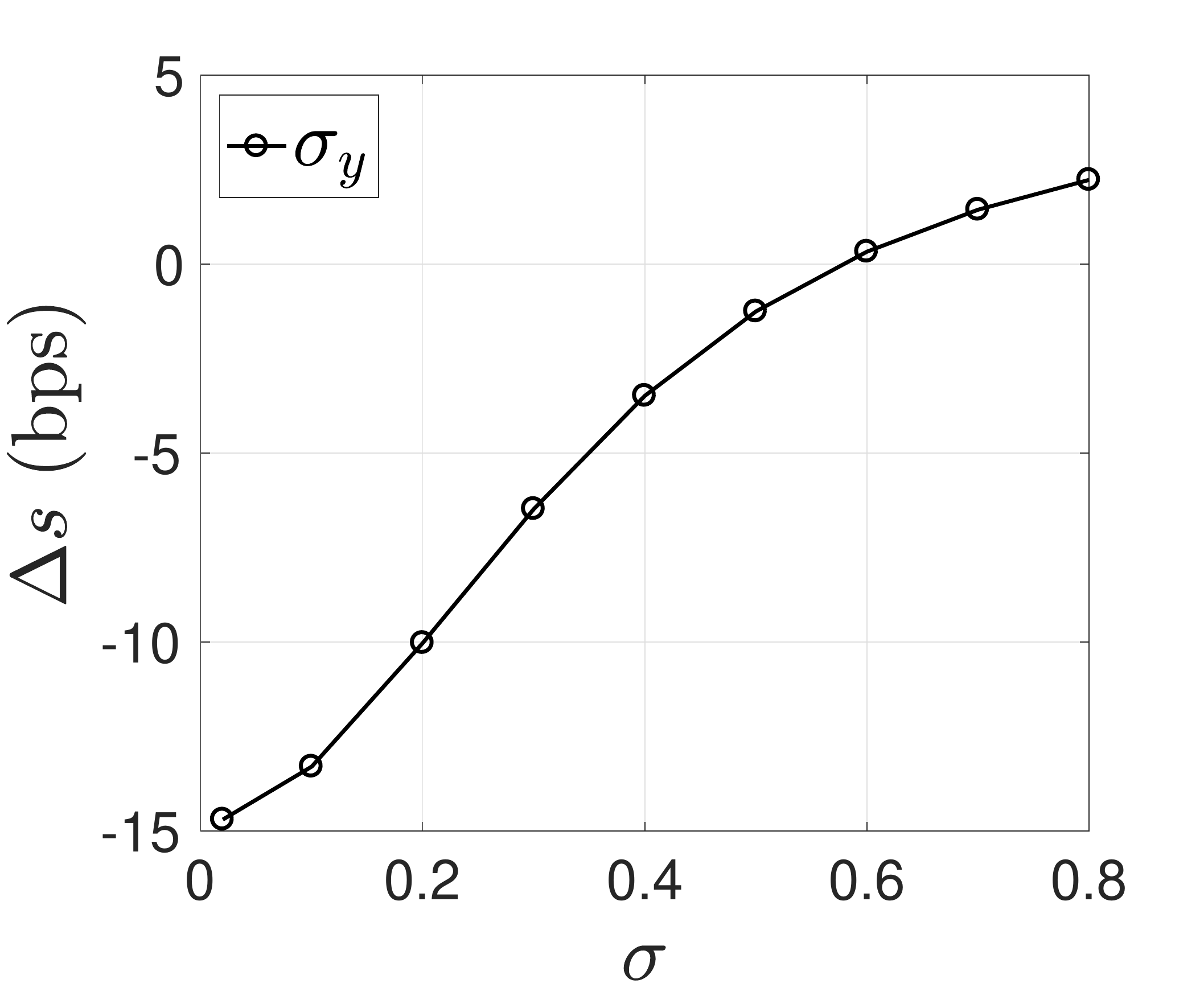}}
     \subfigure{\includegraphics[width=0.32\textwidth]{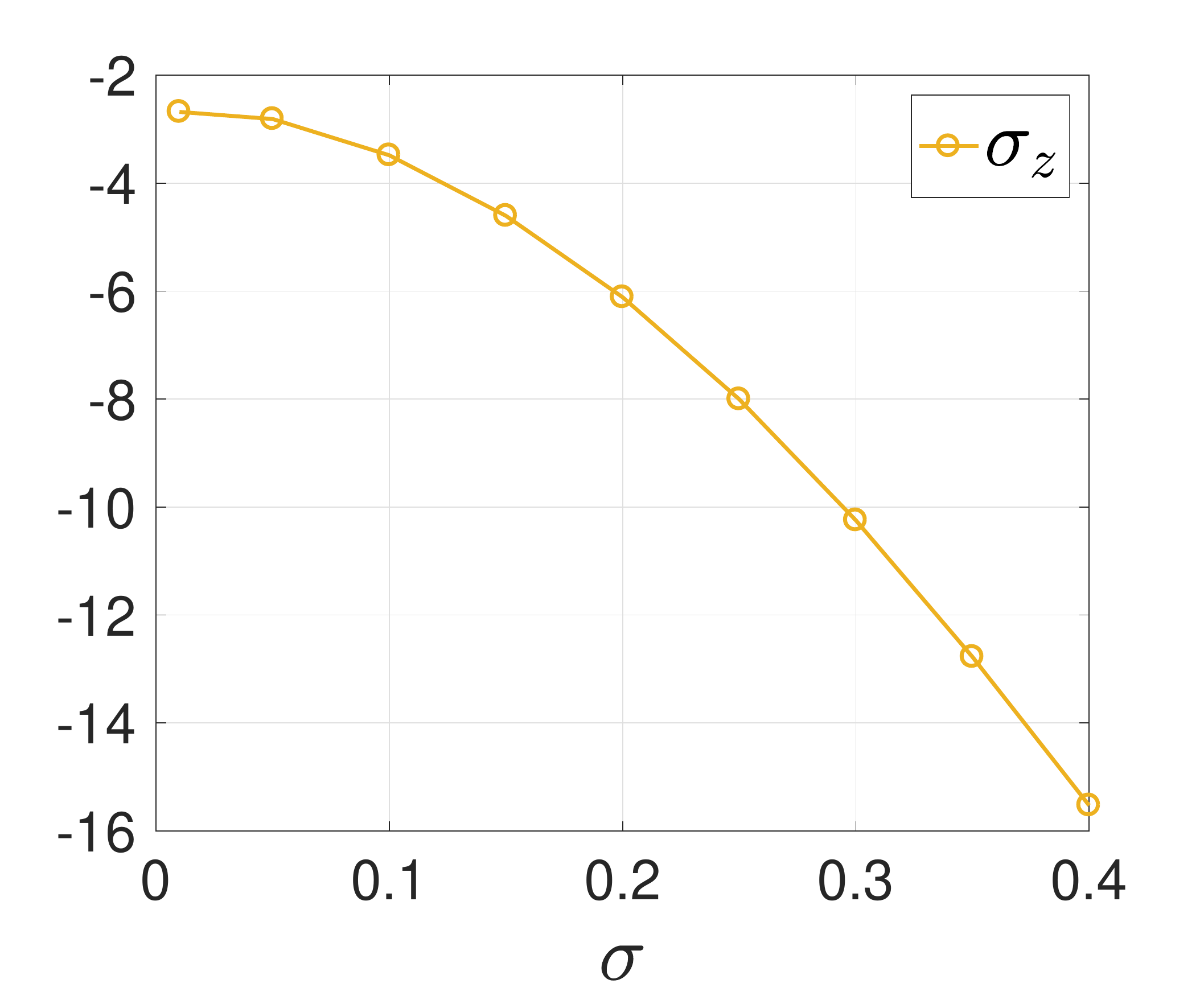}}
     \subfigure{\includegraphics[width=0.32\textwidth]{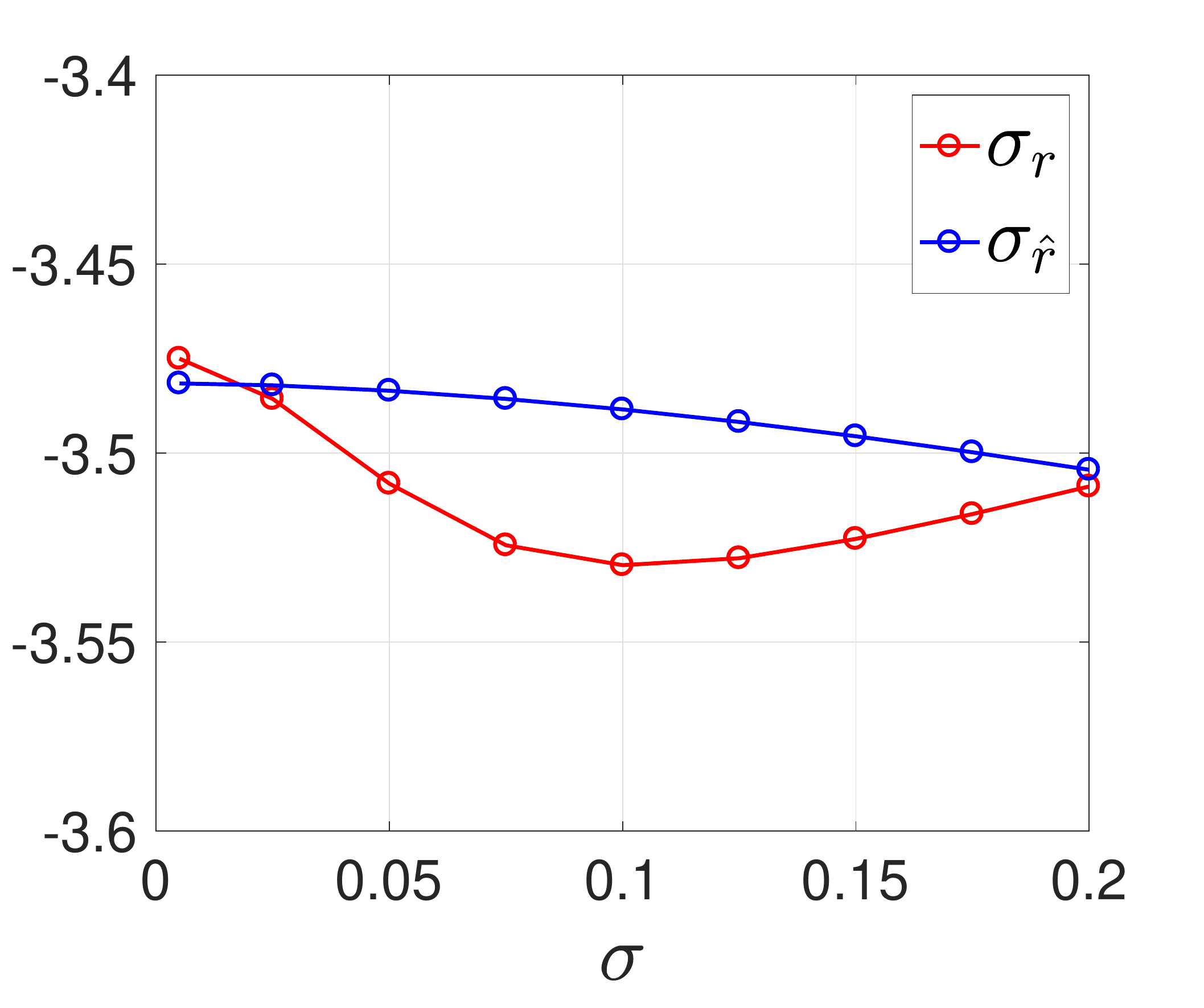}}
\caption{The influence of the volatilities on the 5Y  CDS par spread. No jumps-at-default are assumed.}
\label{fig:Volatility}
\end{figure}
%
%The jump in the FX rate changes the structure of sensitivity of the CDS par spread with respect to different stochastic factors. The jump in the FX rate diminishes the impact of other parameter on the CDS par spread, unlike the jump in the foreign interest rate, which almost does not affect the sensitivity structure and just offsets the curves with the amount of approximately 10 bps per one extra unit in the jump amplitude.

%The impact of the jump in the FX rate on the CDS par spread was well analysed in \cite{Brigo}. We can observe in Figures~\ref{CorrJumpZ}, \ref{VolYJumps}, and \ref{VolZJumps} that the results go in line with that study. Therefore here we focus more on the investigation of the impact coming from the jump in the foreign interest rat

To analyze the impact of the two most influential volatilities under the presence of jumps in $\hatR_t$, we test how the volatility level affects the foreign CDS par spread with respect to the jump amplitude. These results are presented in Fig.~\ref{fig:GamRhazard}, \ref{fig:GamRfx}. Increasing $\sigma_y$ in combination with a $100\%$ raise in $\hatR$ gives rise to the basis spread changing the sign from being negative to positive, while the absolute value of the growth in $\Delta s$ is
about 15 bps. However, the influence of $\sigma_z$ is just opposite. Larger values of $\sigma _z$ give rise to a negative basis spread, which though can be somewhat compensated by the increasing amplitude $\gamma_\hatr$ of the jump-at-default in the foreign interest rate $\hatR_t$.
%Note, that in the lower right panels of Fig.~\ref{fig:GamRhazard}, \ref{fig:GamRfx} we shift the graphs to begin from the same position to show how the volatility affects the structure of the impact of the jump in the foreign interest rate.
%
\begin{figure}[H]
\centering
     \subfigure{\includegraphics[width=0.4\textwidth]{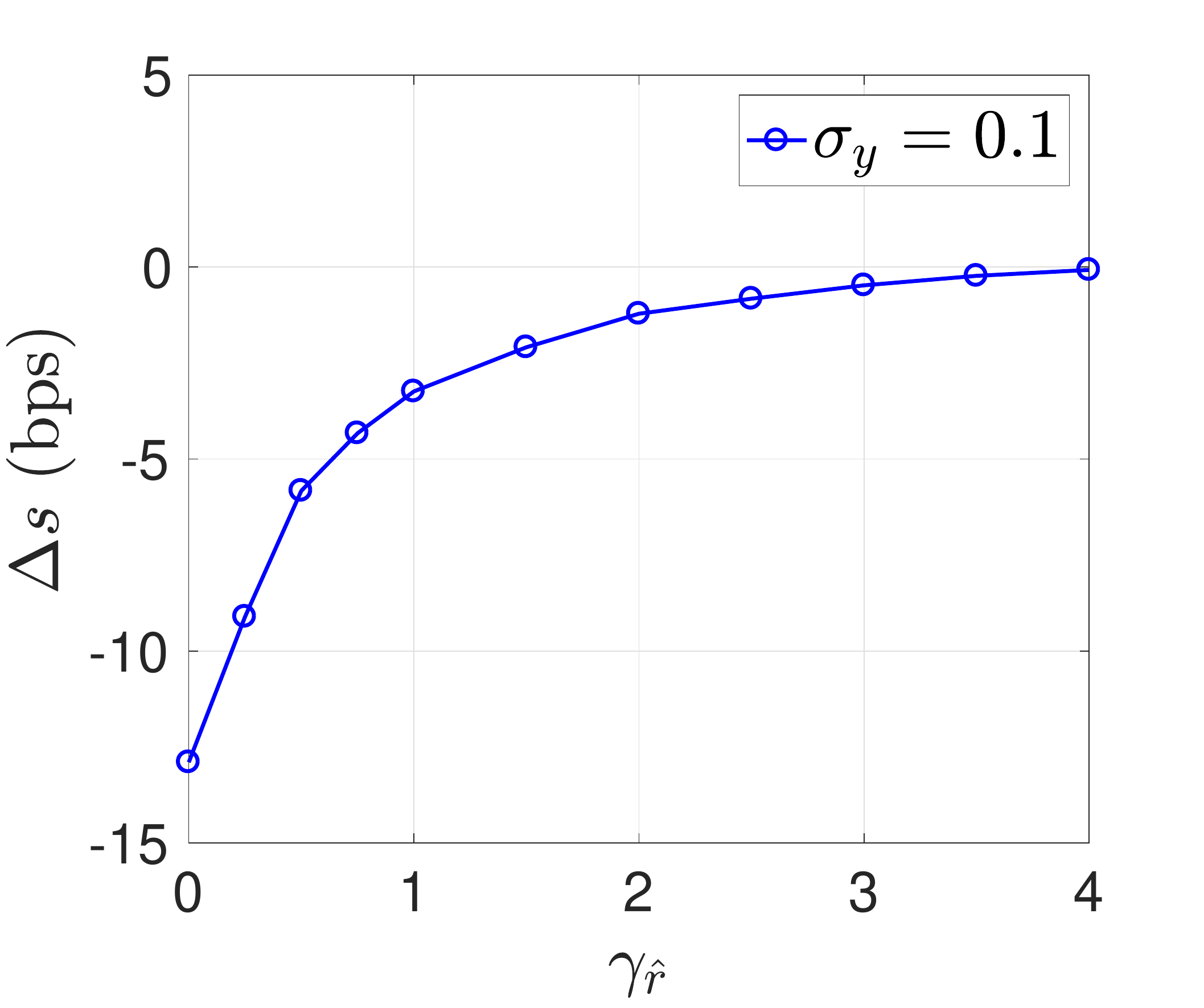}}
     \hspace{0.7cm}
     \subfigure{\includegraphics[width=0.4\textwidth]{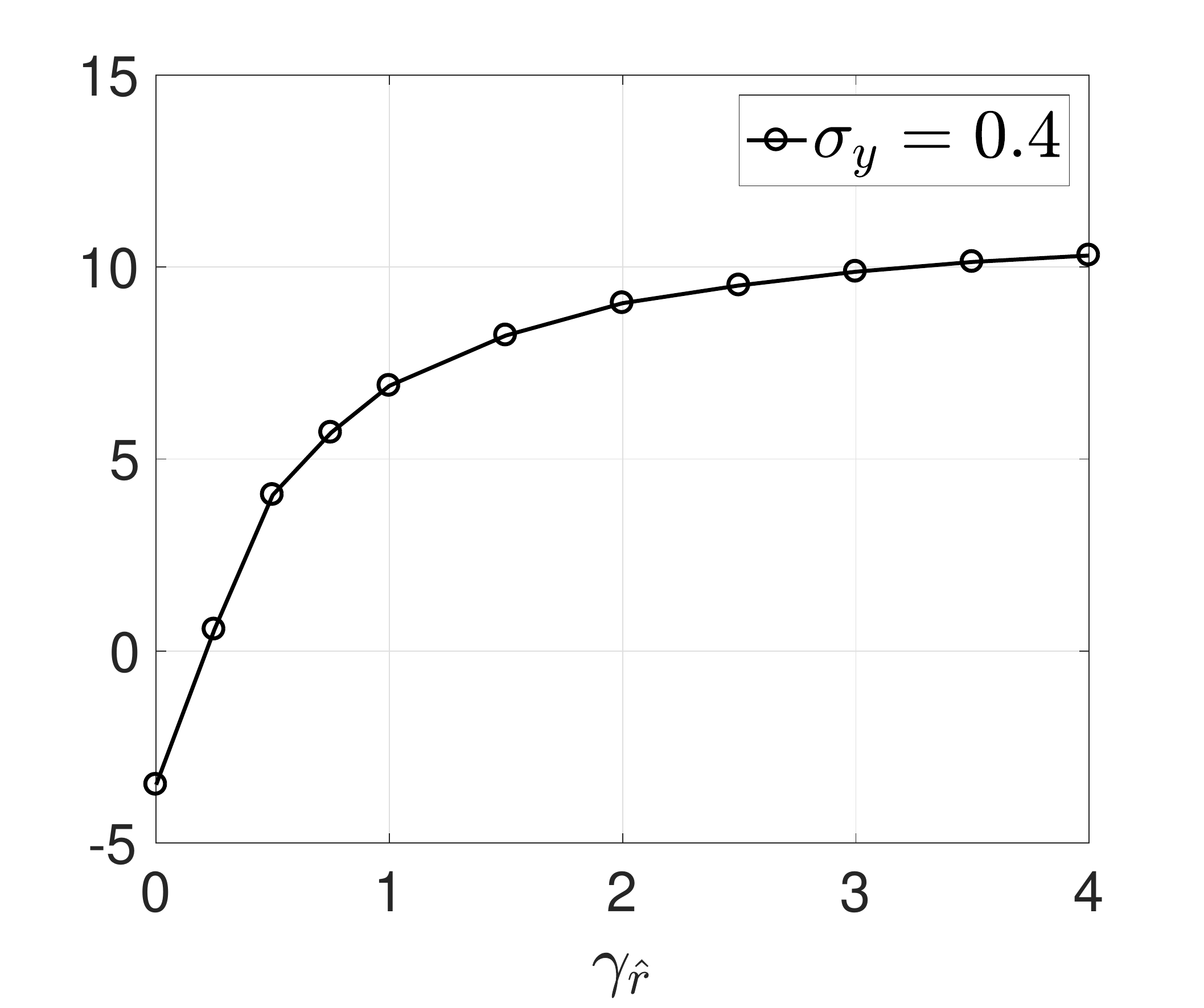}} \\
     \subfigure{\includegraphics[width=0.4\textwidth]{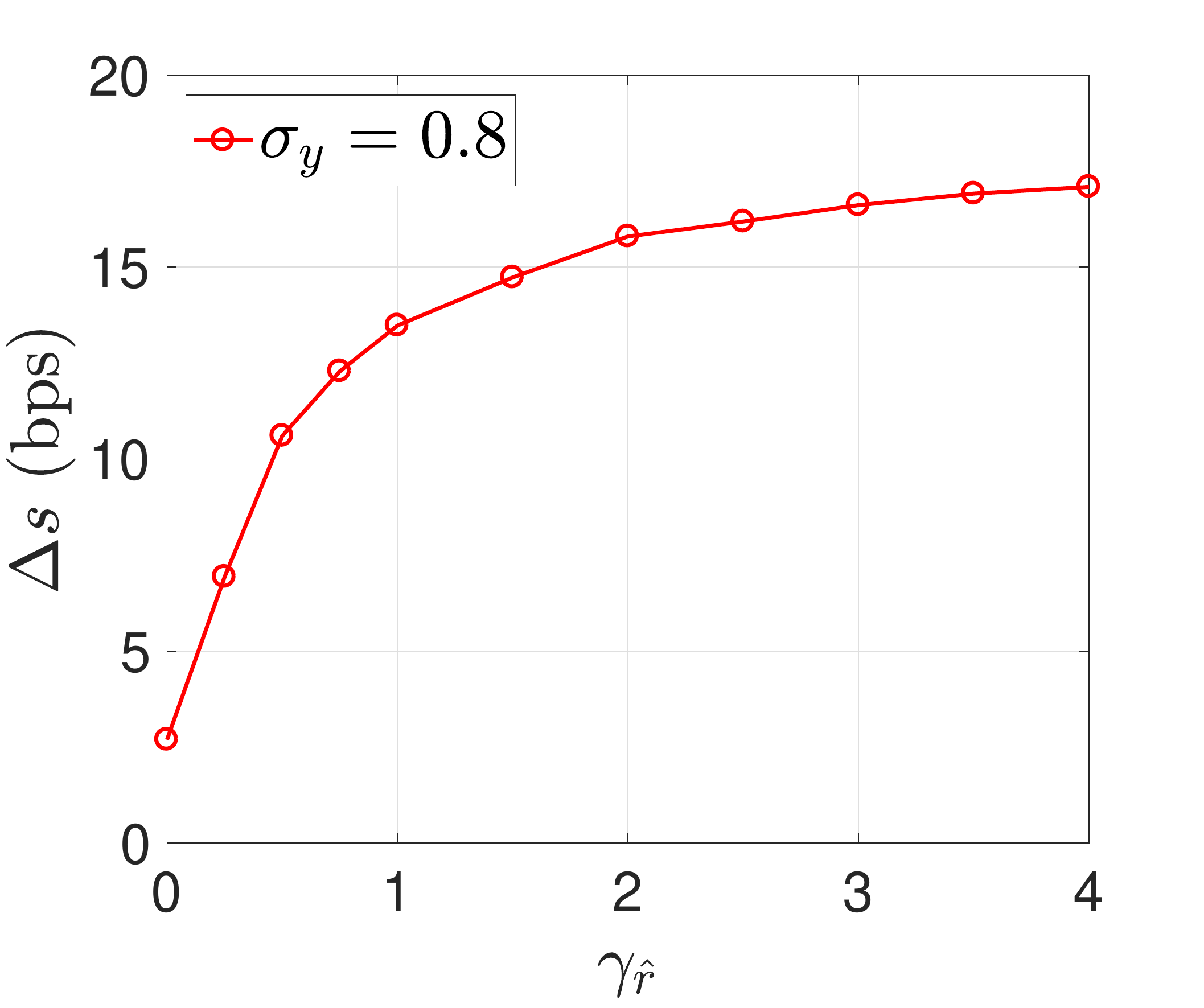}}
     \hspace{0.7cm}
     \subfigure{\includegraphics[width=0.4\textwidth]{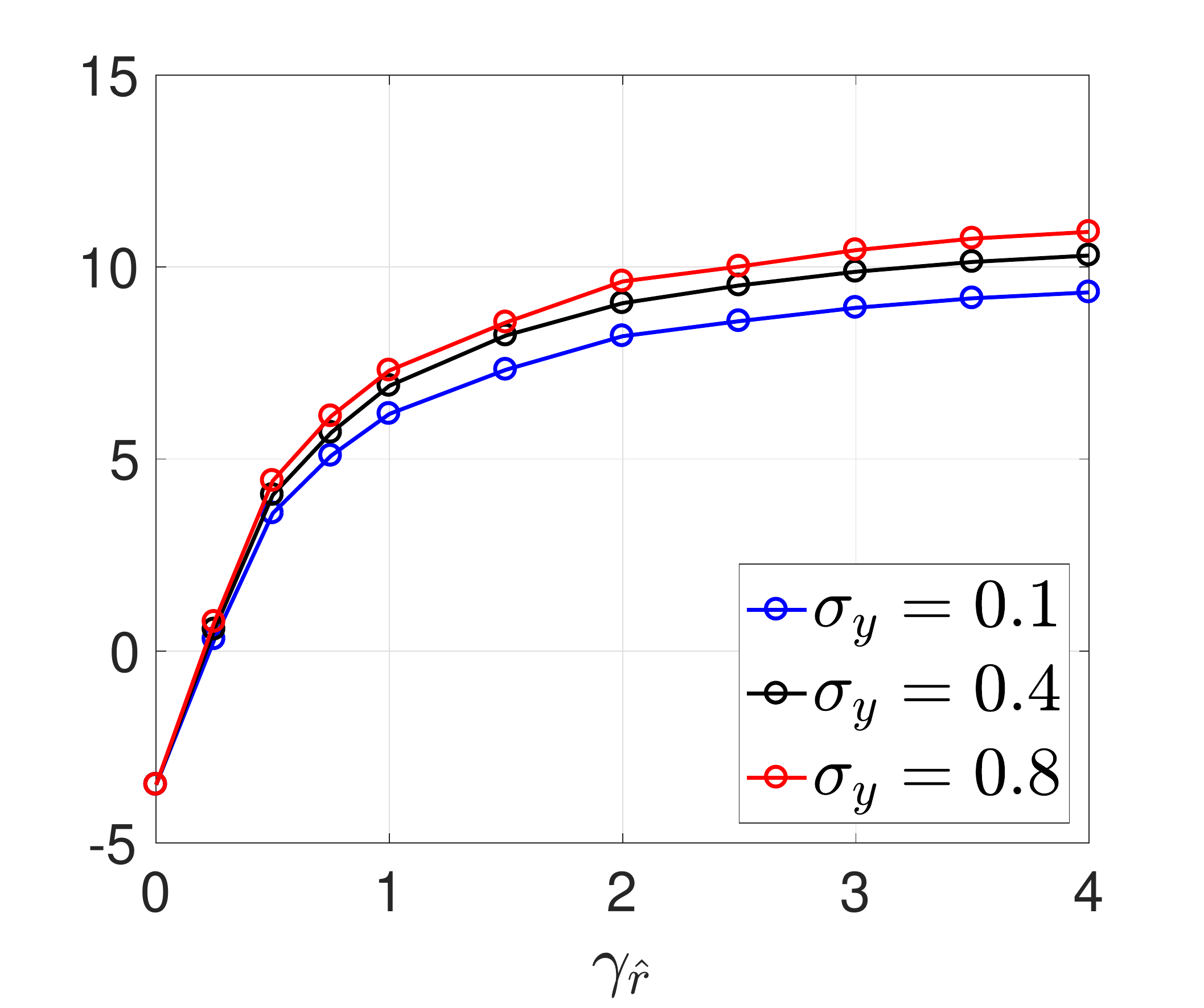}}
\caption{The influence of the hazard rate volatility $\sigma_y$ on the basis spread as a function of $\gamma_\hatr$ at $\gamma_z = 0$. Note, in the lower right panel the lines are shifted to start from the same point.}
\label{fig:GamRhazard}
\end{figure}

\begin{figure}[H]
\centering
     \subfigure{\includegraphics[width=0.4\textwidth]{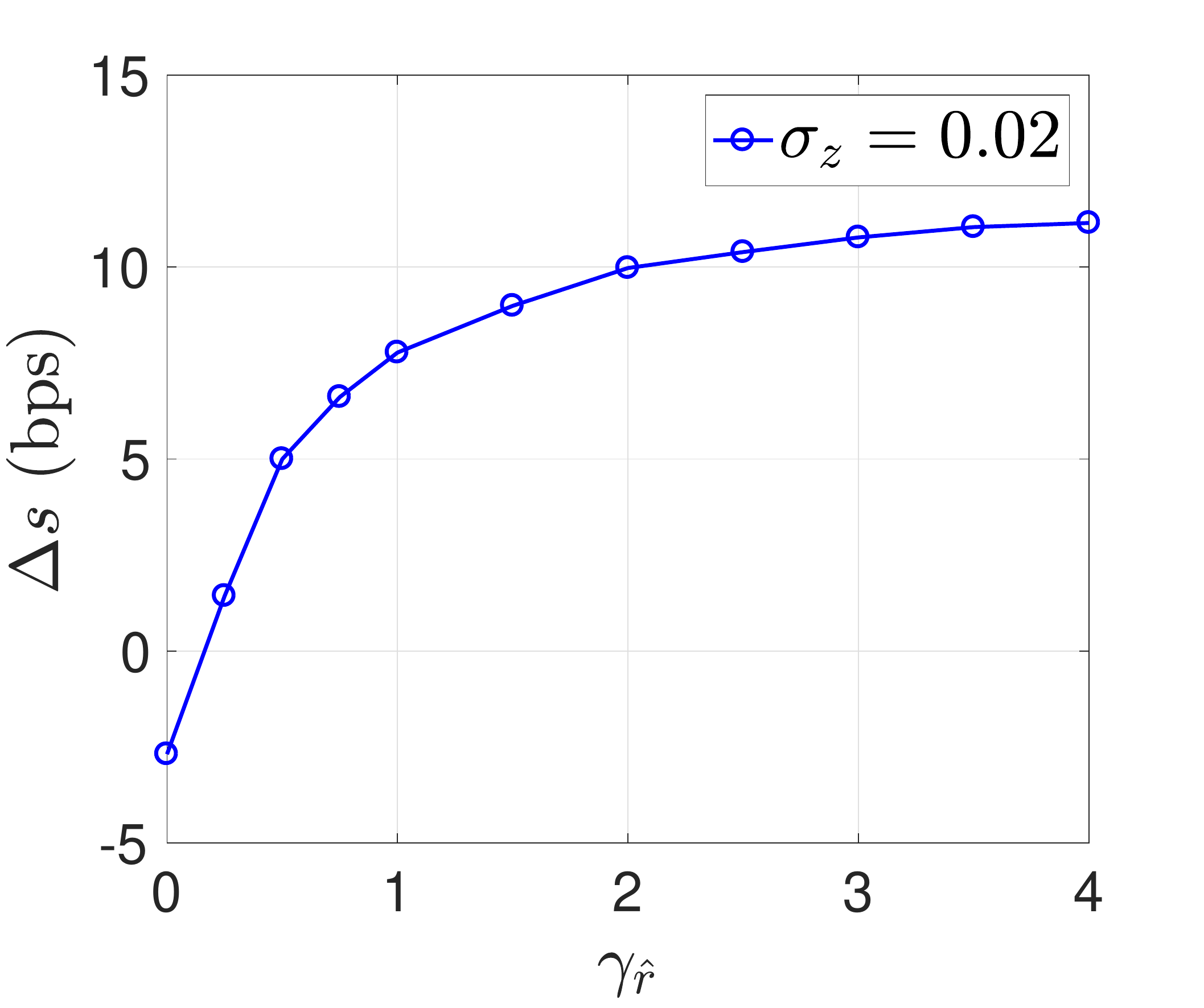}}
     \hspace{0.7cm}
     \subfigure{\includegraphics[width=0.4\textwidth]{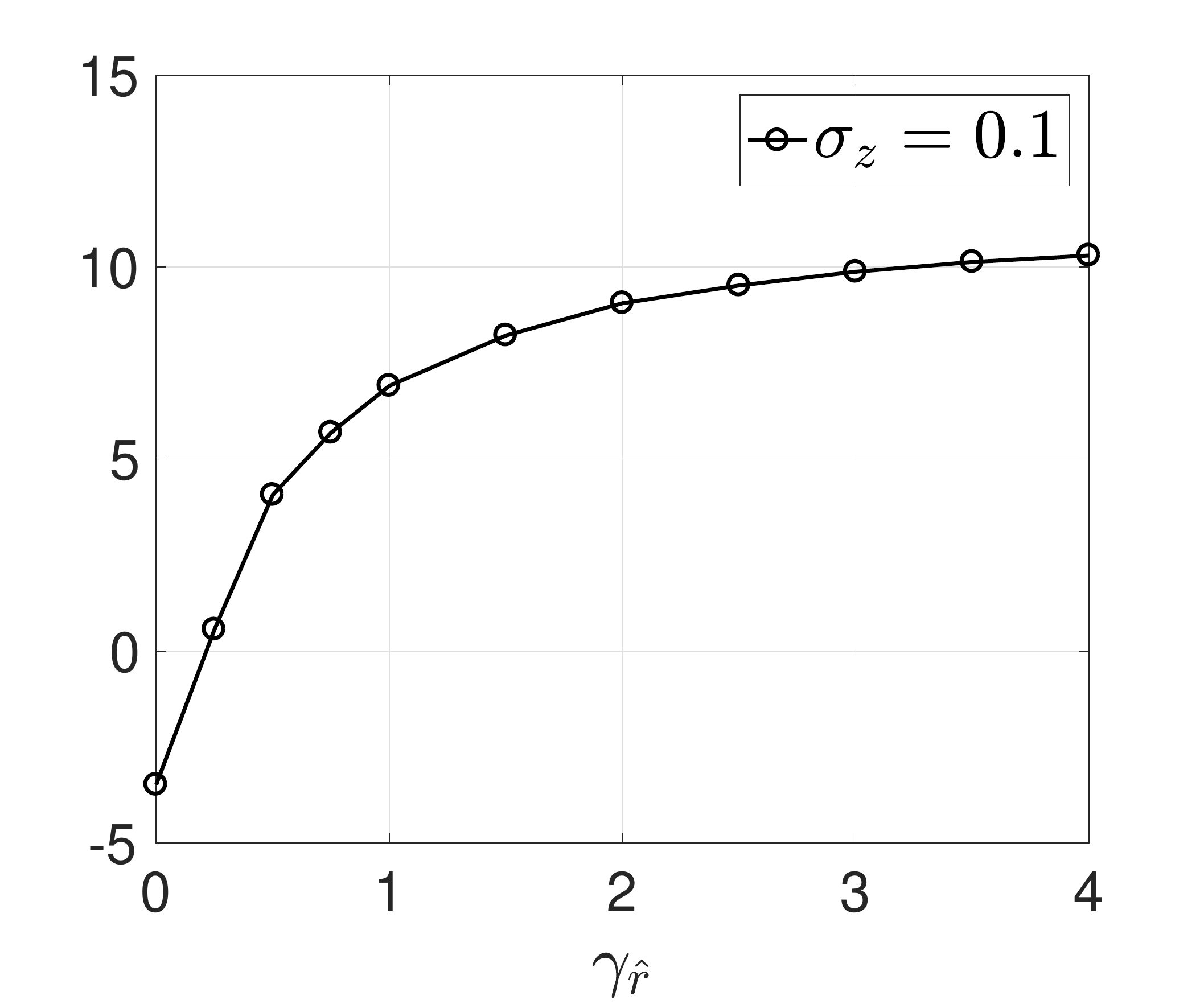}} \\
     \subfigure{\includegraphics[width=0.4\textwidth]{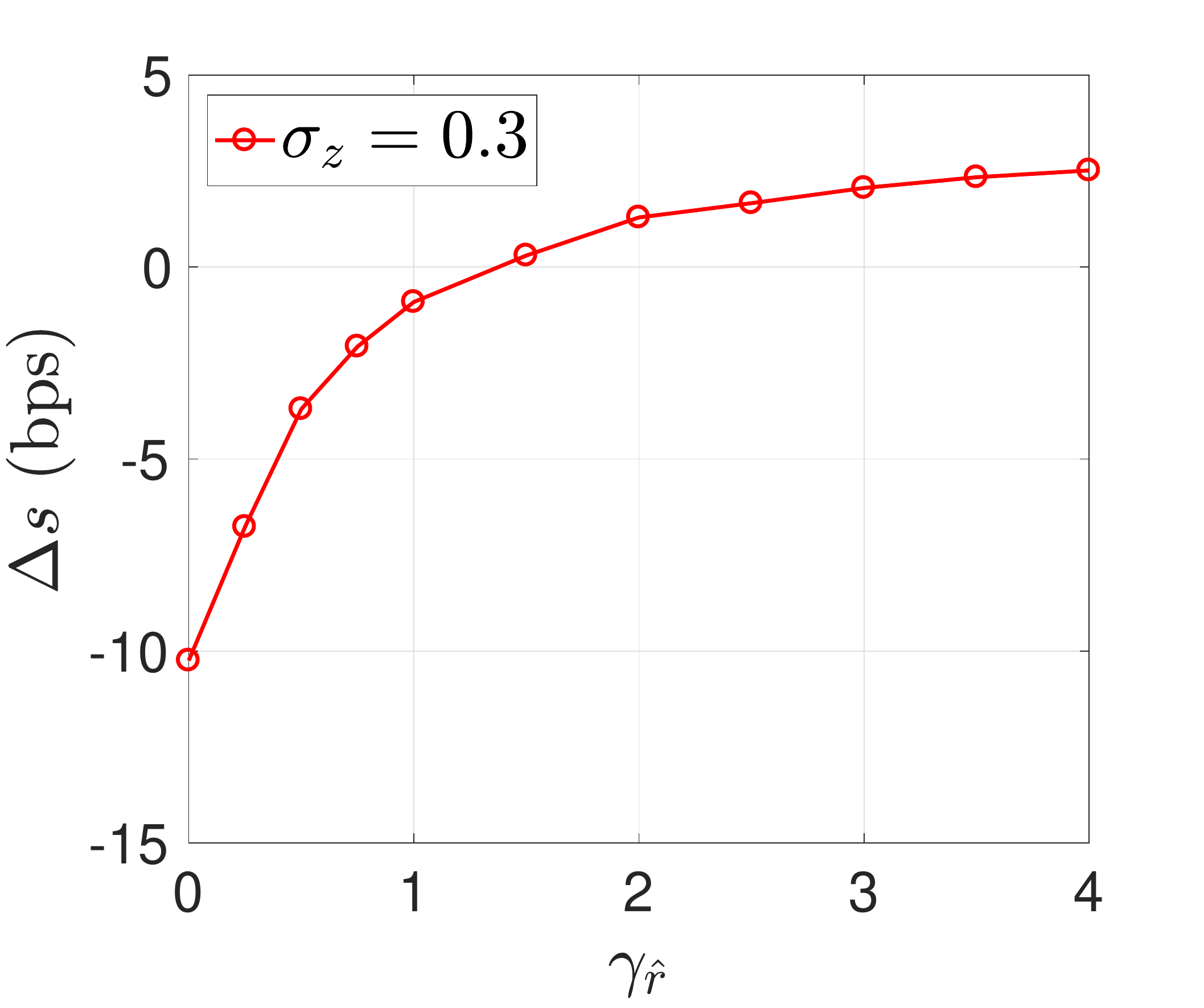}}
     \hspace{0.7cm}
     \subfigure{\includegraphics[width=0.4\textwidth]{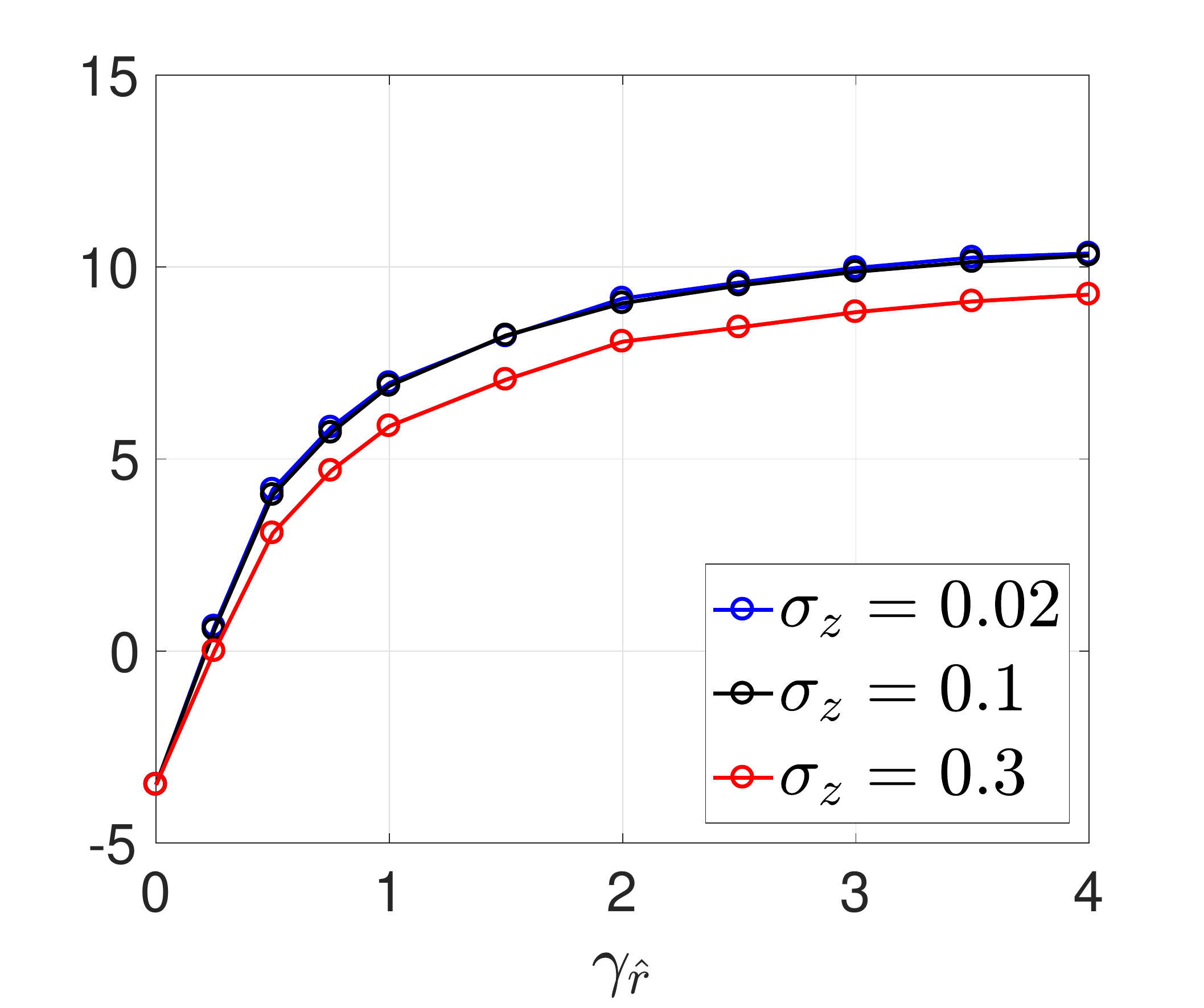}}
\caption{The influence of $\sigma_z$ on the basis spread as a function of $\gamma_\hatr$ at $\gamma_z = 0$. Note, in the lower right panel the lines are shifted to start from the same point.}
\label{fig:GamRfx}
\end{figure}

Thus, we observe that the jump-at-default in the FX rate is the most prominent factor that explains the largest portion of the known discrepancies between Quanto CDS quotes in US dollars and the foreign currency. Nevertheless, the potential jump in the foreign interest rate might be responsible for about 20 bps in the basis spread value. However, it is important to notice that the two jumps have opposite effects: the jump in the FX rate decreases the value of the foreign CDS, while the jump in the IR increases the value of the foreign CDS.

\section{Conclusion}\label{sec:Conclusion}

This paper introduces a new model which can be used, e.g., for pricing Quanto CDS.
The model operates with four stochastic factors, namely, hazard rate, foreign exchange rate, domestic interest rate, and foreign interest rate, and also allows for jumps-at-default in the FX and foreign interest rates. Corresponding systems of PDEs for both the risky bond price and the CDS price is derived similar to how this is done in \cite{BieleckiPDE2005}.

In order to solve these equations we develop a localized radial basis function method that is based on the partition of unity approach. The advantage of the method is that in our four-dimensional case it maintains high accuracy while uses less resources than, for example, corresponding finite difference or Monte Carlo methods. Potentially, the RBF method can be a subject of parallelization which would improve the computational efficiency.

The results of our numerical experiments presented in the paper qualitatively explain the discrepancies observed in the marked values of CDS spreads traded in domestic and foreign
economies and, accordingly, denominated in the domestic (USD) and foreign (euro, ruble, real, etc.) currencies. The Quanto effect (the difference between the prices of the same CDS contract traded in different economies, but represented in the same currency)
can, to a great extent, be explained by the devaluation of the foreign currency. This  would yield a much lower protection payout if converted to the US dollars. These results are similar to those obtained in \cite{Brigo}. We underline, however, that in \cite{Brigo} only constant foreign and domestic interest rates are considered, while in this paper they are stochastic even in the no-jumps framework.

In contrast to \cite{Brigo}, in this paper we also analyze the impact of the jump-at-default in the foreign interest rate which could occur simultaneously with the default in the FX rate. We found that this jump is a significant component of the process and is able to explain about 20 bps of the basis spread value. However, it is worth noticing that the jumps in the FX rate and IR have opposite effects. In other words,  devaluation of the foreign currency will decrease the value of the foreign CDS, while the increase of the foreign interest rate will increase the foreign CDS value.

The other important parameters of the model are correlations between the hazard rate and the factors that incorporate jumps, i.e., $\rho_{yz}$ and $\rho_{y\hatr}$, and volatilities of the hazard process $\sigma_y$ and the FX rate $\sigma_z$.
Therefore, they have to be properly calibrated. Varying the other correlations just slightly contributes to the basis spread value. Large values of the volatilities can in some cases explain up to 15 bps of the basis spread value.

We also have to mention that the pricing problem was formulated via backward PDEs. Therefore, computation of the CDS spread requires to independently solve these PDEs for every discrete time point on a temporal grid lying below the contract maturity.
This could be significantly improved if instead of the backward PDE we would work with the forward one for the corresponding density function. We leave this improvement to be implemented elsewhere.

%%%%%%%%%%%%%%%%%%%%%%%%%%%%%%%%%%%%%%%%%%%%%%%%%%%%%%%%%%%%
\section*{Acknowledgments}
%%%%%%%%%%%%%%%%%%%%%%%%%%%%%%%%%%%%%%%%%%%%%%%%%%%%%%%%%%%%
We thank Peter Carr and Damiano Brigo for their useful comments and discussions.
Victor Shcherbakov acknowledges the support from H F Sederholms stipendiestiftelse, Rektors resebidrag fr\r{a}n Wallenbergstiftelsen, and Anna Maria Lundins stipendiefond. Victor Shcherbakov also thanks the Department of Finance and Risk Engineering at Tandon School, NYU where he worked on this paper as a visiting scholar. We assume full responsibility for any remaining errors.

\clearpage
%%%%%%%%%%%%%%%%%%%%%%%%%%%%%%%%%%%%%%%%%%%%%%%%%%%%%%%%%%%%
\section*{References}
%\bibliographystyle{apalike}

%\bibliography{quantoCDS}

%%%%%%%%%%%%%%%%%%%%%%%%%%%%%%%%%%%%%%%%%%%%%%%%%%%%%%%%%%%%
\clearpage
\begin{appendices}
\renewcommand{\theequation}{\Alph{section}.\arabic{equation}}
\setcounter{equation}{0}
\renewcommand{\theequation}{A.\arabic{equation}}

\section{Derivation of main PDE  \label{apDeriv}}

Below we give a sketch of derivation of the main PDE for the defaultable zero-coupon bond price which follows from our model introduced in Section~\ref{modelJumps}, as the detailed derivation is rather long. Therefore, we utilize some results known in the literature, and explain only the main steps of the derivation.

According to our model setting which is presented in Section~\ref{modelJumps}, all underlying stochastic processes $R_t,\hatR_t, Y_t, Z_t, D_t$ possess a strong Markovian property, see, e.g., \cite{BieleckiPDE2005}. Denote by $r,\hatr, y, z$, and $d$ the initial values of these processes at time $t$, respectively. For Markovian underlyings it is well-known, e.g., \cite{ethier2009markov}, that evolution of $U_t$ represented as a function of variables $(t, r, \hatr, y, z, d)$ can be described by a corresponding PDE (or PIDE if jumps are also taken into account). In this section we derived such a PIDE in the explicit form.

Let us remind that in the jump-at-default framework the dynamics of $Z_t$ and $\hatR_t$ is given by \eqref{dzJump}, \eqref{rJump}
\begin{align} \label{jumpRZ}
dZ_t &=  (R_t - \hatR_t) Z_t dt + \sigma_z Z_t dW_t^{(3)} + \gamma_z Z_t d M_t, \\
d\hatR_t &= \hat a(\hat b-\hatR_t ) dt +  \sigma_{\hatr} \sqrt{\hatR_t}dW_t^{(2)} + \gamma_{\hatr} R_t d D_t. \nonumber
\end{align}
For the sake of convenience the second SDE could be re-written in the form of the first one
\begin{align} \label{hat_R}
d\hatR_t &= \hat a \left(\hat b-\hatR_t \right) dt + d \Gamma_t +  \sigma_{\hatr} \sqrt{\hatR_t}dW_t^{(2)} + \gamma_{\hatr} R_t d M_t \\
&= \hat a \left(\hat b - \hatR_t - \frac{\lambda_t}{\hat a} (1-D_t) \right) dt + d \Gamma_t +  \sigma_{\hatr} \sqrt{\hatR_t}dW_t^{(2)} + \gamma_{\hatr} R_t d M_t. \nonumber
\end{align}
So we replaced $D_t$ with a compensated martingale $M_t$ by subtracting a compensator of $D_t$, and, accordingly, added this compensator to the drift. When doing so, we take into account \eqref{hazard} to obtain $d \Gamma_t = (1 - D_t) \lambda_t dt$.

Below we need the following theorem from \cite{JacodShiryaev:87} (see also \cite{ItkinBook} and references therein), which provides a generalization of It\^{o}'s lemma to the class of semimartingales
\begin{theorem} \label{itoLevy}
Let $X = (X_t)_{0\le t \le T}$ be a \LY process which is a real-valued semimartingale
with the triplet $(b, c, \nu)$, and $f$ be a function on $\mathbb{R}$, $f \in C^2$. Then, $f(X)$ is a semimartingale, and $\forall t \in [0,T]$ the following representation holds
\begin{align} \label{itoLevyForm}
f(X_t) &= f(X_0) + \int_0^t f'(X_{s^-}) d X_s + \dfrac{1}{2} \int_0^t f''(X_{s^-}) d\langle X^c\rangle_s \\
&+ \sum_{0 \le s \le t} \left[ f(X_{s}) - f(X_{s^-}) - f'(X_{s^-}) \Delta X_s \right]. \nonumber
\end{align}
Here $X_{s^-} = \lim_{u \nearrow s}$ is the value just before a potential jump,
$\Delta X_s = X_s - X_{s^-}$, $X^c$ is the continuous martingale part of $X_t$, i.e. $X^c_t = \sqrt{c}W_t$, and $\langle \cdot \rangle$ determines a quadratic variation.

Alternatively, if the random measure of jumps $\mu^X(ds,dx)$ is used, we have
\begin{align} \label{JumpInt}
f(X_t) &= f(X_0) + \int_0^t f'(X_{s^-}) d X_{s^-} + \dfrac{1}{2} \int_0^t f''(X_{s^-}) d\langle X^c\rangle_{s^-} \\
&+ \int^t_0 \int_{\mathbb{R}} \left[ f(X_{s^-} + x) - f(X_{s^-}) - x f'(X_{s^-})\right] \mu^X(ds,dx). \nonumber
\end{align}
\end{theorem}
\begin{proof}
 See Theorem I.4.57 in \cite{JacodShiryaev:87}.
\end{proof}

Further let us consider only jumps of a finite variation and finite activity, so
\[
\sum_{0 \le s \le t} f(X_{s}) < \infty, \qquad
\sum_{0 \le s \le t} f'(X_{s^-}) \Delta X_s < \infty. \]

Our model allows only for a single jump to occur at the default time $\tau$. Therefore,
\begin{equation} \label{measureMu}
\mu^X(ds dx) = \delta(s-t) D_s \nu(dx) ds , \nonumber
\end{equation}
\noindent with $\nu(dx)$ being a \LY measure of jumps in $\mathbb{R}$, and where $\delta(x)$ is the Dirac delta function.

Respectively, in the differential form and for a multidimensional case \eqref{JumpInt}  reads
\begin{align} \label{JumpIntMult}
d f(\Xs) &= \fp{f(\Xsm)}{\Xsm} * d \Xs + \dfrac{1}{2} \sop{f(\Xsm)}{\Xsm} * d\langle \mathbf{X}^c\rangle_s \\
&+ \int_{\mathbb{R}} \left[ f(\Xsm + \mathbf{x}) - f(\Xsm) - \mathbf{x} * \fp{f(\Xsm)}{\Xsm} \right] \nu(d\mathbf{x}) d D_t, \nonumber
\end{align}
\noindent where $\Xs$ is a vector of independent variables, $\mathbf{x}$ is a vector of the corresponding jump values, and $\left<*\right>$ is an inner product.

Also, according to \eqref{jumpRZ} the size of the jump in both the foreign interest rate and the FX rate is proportional to the value of the corresponding process right before the jump occurs at time $\tau$ with a constant rate $\gamma$.

Combining \eqref{measureMu} and \eqref{jumpRZ} gives rise to the \LY measure $\nu(d \mathbf{x})$ of this multi-dimensional jump process to be
\begin{align} \label{levyM}
\nu(d \mathbf{x}) &= \delta(x_z - \gamma_z z)
\delta(x_\hatr - \gamma_\hatr \hatr) d x_z d x_\hatr,
\end{align}
\noindent (compare, e.g., with \cite{Crosby2013}).

Therefore, the last line in \eqref{JumpIntMult} changes to
\begin{align} \label{jump2d}
J &= \left[f(t, \Xs) - f(t, \Xsm) - \Delta \Xsm*\fp{ f(t, \Xsm)} {\Xsm }\right] d D_t, \\
\Xs &= \Xsm + \Delta \Xsm = f(t, r, \hatr(1+\gamma_\hatr), y, z(1+\gamma_z), d=1), \nonumber \\
\Xsm &= f(t, r, \hatr, y, z, d=0). \nonumber
\end{align}
Having all these results, the PDE for the discounted defaultable bond price can be derived by using a standard technique for jump-diffusion processes, see, e.g., \cite{PAPAPANTOLEON2008}. However, for the sake of brevity, we will utilize the approach of \cite{BieleckiPDE2005}, where a similar problem is considered, and, hence, making a reference to the corresponding theorems proved in that paper.

Note, that
\begin{align} \label{dDer}
\mathbb{E}_t[d D_t | D_t & d, Y_t = y] = d \mathbb{E}_t[D_t | D_t = d, Y_t = y] &=  \lambda_t \m1_{t \le \tau} \Big|_{(D_t = d, Y_t = y)} dt = (1 - d) e^y dt,
\end{align}
\noindent where the last by one equality follows from Lemma 7.4.1.3 in \cite{Jeanblanc2009}.

Using \eqref{dDer}, it can be seen that after  the default occurs, $D_t = 1_{\tau \le t} = 1$, and thus the jump term $J$ disappears. However, before the default at time $t < \tau$ the jump term is
\begin{equation} \label{Jfinal}
J = f(t, \Xs) - f(t, \Xsm) - \Delta \Xsm * \fp{ f(t, \Xsm)} {\Xsm }.
\end{equation}

So, conditional on the value of $D_t$ the solution could be represented in the form
\begin{equation}
f(t, \Xs) = \m1_{t < \tau} f(t, \Xsm) + \m1_{\tau \le t} f(t, \Xs).
\end{equation}
Then the remaining derivation of the PDE could be done based on the following Proposition:
\begin{proposition}[Proposition~3.1 in \cite{BieleckiPDE2005}] \label{Prop}
Let the price processes $Y^i, \ i=1,2,3$ satisfy
\[ d Y^i_t = Y^i_{t^-} \left[ \mu_i dt + \sigma_i d W^i_t + k_i d M_t\right] \]
\noindent with $k_i > -1$ for $i = 1,2,3$, $\mu, \sigma$ being the corresponding drifts and volatilities. Then the arbitrage price of a contingent claim $Y$ with the
terminal payoff $G(t, Y^1_T, Y^2_T, Y^3_T, D_T)$ equals
\[  \pi_t(Y) = \m1_{t \le \tau} C(t, Y^1_t, Y^2_t, Y^3_t, 0) + \m1_{t \ge \tau} C(t, Y^1_t, Y^2_t, Y^3_t, 1) \]
\noindent for some function $C: [0,T] \times  \mathbb{R}^3_+ \times {0,1} \to \mathbb{R}$. Assume that for $d = 0$ and $d=1$ the auxiliary function
$C(\cdot,d): [0,T] \times  \mathbb{R}^3_+ \to \mathbb{R}$ belongs to the class
$C^{1,2}([0,T] \times  \mathbb{R}^3_+)$. Then the functions $C(\cdot,0)$ and $C(\cdot,1)$
solve the following PDEs:
\begin{align*}
\p_t C(\cdot,0) &+ \sum_{i=1}^3 (\alpha - \lambda k_i)y_i \p_i C(\cdot,0) +
\frac{1}{2} \sum_{i,j=1}^3 \rho_{ij} \sigma_i \sigma_j y_i y_j \p_{ij} C(\cdot,0) \\
&+ \lambda \left[ C(t, y_1(1+k_1), y_2(1+k_2), y_3(1+k_3), 1)  - C(t, y_1, y_2, y_3, 0)\right] - \alpha C(\cdot,0) = 0, \\
\p_t C(\cdot,1) &+ \sum_{i=1}^3 \alpha y_i \p_i C(\cdot,1) +
\frac{1}{2} \sum_{i,j=1}^3 \rho_{ij} \sigma_i \sigma_j y_i y_j \p_{ij} C(\cdot,1)
- \alpha C(\cdot,1) = 0.
\end{align*}
\noindent subject to the terminal conditions
\begin{align*}
C(T, y_1, y_2, y_3, 0) &= G((T, y_1, y_2, y_3, 0), \\
C(T, y_1, y_2, y_3, 1) &= G((T, y_1, y_2, y_3, 1). \\
\end{align*}

\end{proposition}
\begin{proof}
See \cite{BieleckiPDE2005}.
\end{proof}
Two important notes should be made in order to apply this proposition to our problem.

\paragraph{Tradable assets} In \cite{BieleckiPDE2005} all underlying assets are assumed to be tradable. Therefore, they have to be martingales under some unique martingale measure (a particular choice of $Y^1$ as is made to be a numeraire). To achieve this, additional conditions on the drifts, volatilities and the jump rates $k_i$ should be imposed. In particular, this requires that the coefficient $\alpha$ in Proposition~\ref{Prop} would be
\[ \alpha = \mu_i + \sigma_i \frac{c}{a}, \]
\noindent where the determinants $c,a$ are the explicit functions of $\mu_i, \sigma_i, k_i, \ i=1,2,3$ and given in \cite{BieleckiPDE2005}. Moreover, it is shown that the right-hand side of this formula does not depend on $i$.

However, for our problem among all the underlying processes the only tradable one is that for the FX rate. This allows one to fully eliminate these conditions on $\mu_i, \sigma_i, k_i, \ i=1,2,3$. As a consequence, e.g., the term
\[ \sum_{i=1}^3 (\alpha - \lambda k_i)y_i \p_i C(\cdot,0) \]
\noindent in the Proposition~\ref{Prop} is now replaced with
\[ \sum_{i=1}^3 (\mu_i - \lambda k_i)y_i \p_i C(\cdot,0). \]

\paragraph{Risk-neutrality} Proposition~\ref{Prop} derives an arbitrage price (under real measure) of the contingent claim written on the given underlyings. To get this price under a risk-neutral measure $\QM$, one needs to construct a replication ({\it self-financing}) strategy of a generic claim. In particular, to hedge out the risk of $\hat R_t$ and $R_t$, corresponding non-defaultable zero-coupon bonds (perhaps, of a longer maturity) should be used as a hedge, \cite{Bielecki2004, WIlmott1998}.

This problem is solved by Proposition~3.3 of \cite{BieleckiPDE2005}. Accordingly, the previously derived PDEs remain the same, with the only change of the killing term where the coefficient $\alpha$ is replaced with the interest $r$ corresponding to measure $\QM$ (as expected based on a general theory of asset pricing).

We proceed by combining these results together and applying them to our model. First, we revert the notation back to that used in this paper. Then, taking into account an explicit form of the stochastic differential equations describing the dynamics of our underlying processes, and conditioning on $R_t = r, \hatR_t = \hatr, Z_t = z, Y_t = y, D_t = d$, we obtain that under the risk-neutral measure $\QM$ the price $U_t(T)$ is
\begin{equation} \label{bondPriceA}
U_t(T, r, \hatr, y, z) = \m1_{t < \tau} f(t, T, r,\hatr, y, z, 0) +
\m1_{t \ge \tau} f(t, T, r,\hatr, y, z, 1).
\end{equation}

Here the function $f(t, T, r,\hatr, y, z, 1) \equiv u(t, T, X), \ X = \{t,r,\hatr, y, z\}$ solves the PDE
\begin{equation} \label{PDE1A}
\fp{u(t,T,X)}{t} + {\cal L} u(t,T,X) - r u(t,T,X) = 0,
\end{equation}
\noindent where the diffusion operator $\cal L$ reads
\begin{align} \label{LdiffA}
\cal L &= \frac{1}{2}\sigma{_r}^2 r\sop{}{r} + \frac{1}{2} \sigma_{\hatr}^2 \hatr \sop{u}{\hatr} + \frac{1}{2}\sigma_z^2 z^2 \sop{}{z} + \frac{1}{2}\sigma_y^2\sop{}{y}
+ \rho_{r \hatr} \sigma_r \sigma_{\hatr} \sqrt{r \hatr}\cp{}{r}{\hatr}  \\
&+ \rho_{rz}\sigma_r \sigma_z z\sqrt{r} \cp{}{r}{z}
+ \rho_{\hatr z} \sigma_{\hatr} \sigma_z z \sqrt{\hatr} \cp{}{z}{\hatr}
+ \rho_{ry}\sigma_r \sigma_y \sqrt{r} \cp{}{r}{y}
+ \rho_{\hatr y} \sigma_{\hatr} \sigma_y \sqrt{\hatr} \cp{}{y}{\hatr}
\nonumber \\
&+ \rho_{yz} \sigma_y \sigma_z z \cp{}{y}{z}
+ a(b-r)\fp{}{r}
+ \hat a(\hat b - \hatr) \fp{}{\hatr}
+ (r - \hatr) z \fp{}{z}
+ \kappa(\theta - y) \p{}{y}. \nonumber
\end{align}

The second function $f(t, T, r,\hatr, y, z, 0) \equiv v(t, T, X)$ solves the PDE
\begin{align} \label{PDE2A}
\fp{v(t,T,X)}{t} &+ {\cal L} v(t,T,X) - r v(t,T,X)
- \lambda \gamma_z z \fp{v(t,T,X)}{z} \\
&+ \lambda \left[ u(t, T, X^+) -
v(t, T, X) \right] = 0, \qquad X^+ = \{r, \hatr(1+\gamma_\hatr), y, z(1+\gamma_z)\}, \nonumber
\end{align}
\noindent where according to \eqref{lambda}, $\lambda = e^y$. Note, that the term
$\lambda \gamma_\hatr \hatr v_\hatr(t,T,X)$ in the drift of \eqref{hat_R} cancels out with the corresponding compensator in \eqref{Jfinal} as it should be as the process $\hatR_t$ is not a martingale.

\end{appendices}

\end{document}